\documentclass[letterpaper,11pt]{article}
\usepackage[top=1in,bottom=1in,left=1in,right=1in]{geometry}
\usepackage[utf8]{inputenc}  

\usepackage[T1]{fontenc}
\usepackage[colorinlistoftodos,bordercolor=orange,backgroundcolor=orange!20,linecolor=orange,textsize=normalsize]{todonotes}
\usepackage{amsmath}
\usepackage{amssymb}
\usepackage{amsthm}
\usepackage{xcolor}
\usepackage{bbm}
\usepackage{accents}
\usepackage{complexity}
\usepackage{booktabs}
\usepackage{bm}
\usepackage{paralist}
\usepackage{fixmath}
\usepackage{tcolorbox}
\usepackage{csquotes}
\usepackage{hyperref}
\usepackage[capitalize]{cleveref}
\usepackage{thm-restate}
\usepackage{enumitem}

\urldef{\marcinhomepage}\url{https://www.mimuw.edu.pl/~malcin/}
\urldef{\szymonhomepage}\url{https://www.mimuw.edu.pl/~szymtor/}

\definecolor{amber}{rgb}{1.0, 0.75, 0.0}

\theoremstyle{plain}
\newtheorem{theorem}{Theorem}
\newtheorem{lemma}[theorem]{Lemma}
\newtheorem{corollary}[theorem]{Corollary}
\newtheorem{claim}[theorem]{Claim}
\newenvironment{claimproof}
  {\begin{proof}[Proof of claim]}
  {\end{proof}}

\makeatletter

\usepackage{xspace}
\usepackage{tikz}
\usepackage[ruled,vlined,linesnumbered]{algorithm2e}

\usetikzlibrary{fit}
\usetikzlibrary{arrows,arrows.meta}
\usetikzlibrary{patterns}
\usetikzlibrary{calc}
\usetikzlibrary{shapes}
\usetikzlibrary{positioning}
\usetikzlibrary{math}
\usetikzlibrary{shapes.symbols}
\usetikzlibrary{decorations.pathreplacing,calligraphy}
\usetikzlibrary{decorations.pathmorphing}
\usepackage[scr=boondox,scrscaled=1.05]{mathalfa}
\usetikzlibrary{shapes.geometric}

\newcommand{\sapsp}{\textsc{APSP}\xspace}

\renewcommand{\geq}{\geqslant}
\renewcommand{\leq}{\leqslant}
\renewcommand{\preceq}{\preccurlyeq}

\renewcommand{\le}{\leq}
\renewcommand{\ge}{\geq}
\renewcommand{\cal}{\mathcal}
\newcommand{\set}[1]{\left\{#1\right\}}

\newcommand{\eps}{\varepsilon}

\crefname{claim}{Claim}{Claims}

\newcommand\abs[1]{\lvert #1\rvert}

\newcommand\sd{\text{sd}}

\newcommand{\floor}[1]{\lfloor #1 \rfloor}

\newcommand{\Nn}{\mathbb{N}}

\newcommand{\Cc}{\mathcal{C}}

\newcommand{\Oh}{\mathcal{O}}

\newcommand{\mw}{\mathsf{mw}}

\newcommand{\res}{\mathbin{\downharpoonright}_}

\newcommand{\trF}{ {\ensuremath{\widehat{F}}} }
\newcommand{\pairs}{ {\ensuremath{\mathcal{P}}} }

\title{Time-Optimal APSP and Matrix Multiplication\\in Classes of Linear Neighborhood Complexity}
\author{
\'{E}douard Bonnet\thanks{CNRS, ENS de Lyon, Université Claude Bernard Lyon 1, LIP UMR 5668, 69342 Lyon, France. Homepage: \url{http://perso.ens-lyon.fr/edouard.bonnet}. Email: \texttt{edouard.bonnet@ens-lyon.fr}.}
\and
Julien Duron\thanks{Institute of Informatics, University of Warsaw, Poland. Email: \texttt{j.duron@uw.edu.pl}}
\and
Marcin Pilipczuk\thanks{Institute of Informatics, University of Warsaw, Poland. Homepage: \marcinhomepage.
Email: \texttt{m.pilipczuk@uw.edu.pl}.}
\and
Marek Sokołowski\thanks{Max Planck Institute for Informatics, Saarland Informatics Campus, Saarbr\"{u}cken, Germany. Homepage: \url{https://mnbvmar.github.io}. Supported by the Deutsche Forschungsgemeinschaft (DFG, German Research Foundation) grant number 559177164. Email: \texttt{msokolow@mpi-inf.mpg.de}.}
\and
Szymon Toruńczyk\thanks{Institute of Informatics, University of Warsaw, Poland. Homepage: \szymonhomepage. Email: \texttt{szymtor@uw.edu.pl}.}
}
\date{\today}

\begin{document}
\begin{titlepage}
\maketitle

\begin{abstract}
  The notion of \emph{linear neighborhood complexity} is a~very general 
  structural assumption on a~graph class, covering most classes of sparse graphs
  such as planar graphs, graphs excluding a~fixed (topological) minor, or bounded expansion graphs,
  as well as many structured classes of dense graphs, such as graphs of bounded clique-width,
  twin-width, merge-width, or flip-width. 

  In this work, we present $\Oh(n^2)$-time optimal algorithms for $n$-vertex graphs 
  coming from a~class of linear neighborhood complexity for the following problems:
  \begin{itemize}
  \item \textsc{All-Pairs Shortest Paths},
  \item the multiplication of the adjacency matrix $M$ of the input graph with any $n \times n$ matrix.
  More specifically, after a quadratic preprocessing, we can multiply~$M$ with any $n$-vector in $\Oh(n)$ time.
  \end{itemize}

  This solves several questions raised in [Bonnet, Kim, Geniet, Moon; ICALP '26], and improves and generalizes results in several other recent papers [Bonnet, Giocanti, Ossona de Mendez, Thomassé; STACS '23], [Bannach, Marwitz, Tantau; STACS '24], [Anand, van den Brand, McCarty; NeurIPS~'26], [Kozma, Opler '26], and [Cardinal, McCarty, Yuditsky '26].
  We also extend our results to classes of bounded VC density. 

  In classes of linear neighborhood complexity, we also give a~triangle-detection algorithm in randomized linear time $\Oh(n+m)$ in $n$-vertex $m$-edge graphs, a~$K_4$-detection algorithm in randomized $\Oh(n \log^5 n + m \log n)$ or deterministic $\Oh(n^2)$ time, and a~$K_5$-detection algorithm in randomized $\Oh(n \log^9 n + m \log^5 n)$ time.  
\end{abstract}

\vfill

	\thispagestyle{empty}
\end{titlepage}

\tableofcontents
\thispagestyle{empty}
\newpage\setcounter{page}{1}

\section{Introduction}\label{sec:intro}

In recent years, we have witnessed tremendous progress in understanding
very general structured graph classes and their algorithmic properties.

Classic studies of structured graph classes tackle topologically constrained
graphs such as planar graphs, or, more generally, graph classes excluding a~minor~\cite{Lovasz06}.
This is probably the limit of improved tractability
for some algorithmic problems, e.g., network design problems, where topological restrictions
on the input graph seem necessary.
However, for problems
that treat unweighted graphs and interactions between vertices
only within a bounded distance, we expect to find efficient algorithms in more general graph classes.
Problems of this type include NP-hard problems such as \textsc{Minimum Dominating Set}
or \textsc{Maximum Independent Set} (where one can hope for
approximation or parameterized algorithms that are believed to be impossible for general graphs)
and polynomial-time solvable problems
such as multiple shortest path computations
(where one can hope to beat known lower bounds that hold in general graphs).

Two decades ago, Ne\v set\v ril and Ossona de Mendez~\cite{NesetrilOdM08a,NesetrilOdM08b,NesetrilOdM10,NesetrilOdM11} initiated the study of \emph{bounded expansion} and \emph{nowhere dense} graph classes.
Subsequent research confirmed that the latter are the limit of tractability of first-order model
checking (a~meta-problem capturing most problems that concern bounded-distance interactions between vertices)
among subgraph-closed graph classes~\cite{GroheKS17}.
However, the assumption of being subgraph-closed does not allow this theory
to capture \emph{dense} structured graph classes.
In attempts to understand the limit of tractability of first-order model checking on graphs
in full generality,
width parameters twin-width~\cite{twin-width1}, flip-width~\cite{Torunczyk23}, and merge-width~\cite{merge-width} were introduced.
Boundedness of the last two is conjectured to be equivalent;  furthermore, tractability of the first-order model checking problem
in hereditary graph classes
is conjectured to be delimited by the slightly more general notions of \emph{almost bounded flip-width} or \emph{almost bounded merge-width} \cite{merge-width}.
Research on algorithms for classes of bounded twin-width/merge-width/flip-width
includes not only parameterized algorithms for first-order model checking~\cite{twin-width1,GajarskyPPT22,merge-width}, 
but also parameterized algorithms for specific problems~\cite{twin-width3,twin-width5,GanianPSSS22,BalabanMR25}, approximation algorithms \cite{twin-width3,BergeBDW23,Drabik26}, polynomial kernels~\cite{BonnetKRTW22}, and polynomial-time algorithms~\cite{twin-width3,KratschNS22,GroheN26,BonnetKGM26,DreierK26}.

However, most of the above-mentioned algorithms require a~witness of low twin-width or merge-width to be given as part of the input.
The existence of polynomial-time or parameterized algorithms providing appropriate witnesses remains a central open question in the area.

First-order formulas are able to speak about any fixed distance between vertices, 
but some problems are concerned only with direct adjacencies.
This makes them potentially efficiently solvable on even larger classes of graphs.
A promising concept in this direction is the notion of \emph{linear neighborhood complexity}~\cite{Gajarsky17}.
We say that a graph class $\mathcal{C}$ has \emph{linear neighborhood complexity}
if there exists a constant $c > 0$ such that for every $G \in \mathcal{C}$ and $\emptyset \neq A \subseteq V(G)$
the number of neighborhood traces on $A$, i.e., $\abs{\{N(v) \cap A~|~v \in V(G)\}}$, is at most $c|A|$.
We also say that $\mathcal{C}$ has VC density~$\rho$ if there exists a constant $c > 0$ such that for every $G \in \mathcal{C}$ and $\emptyset \neq A \subseteq V(G)$, $\abs{\{N(v) \cap A~|~v \in V(G)\}} \leqslant c|A|^\rho$.

Linear neighborhood complexity not only covers all classes of bounded twin-width, flip-width, and merge-width but also, for example, 
classes with bounded \emph{weak $2$-coloring number}, such as the class of 2-subdivisions of all graphs.
An even more general notion---extending graph \emph{degeneracy} to the dense setting---is bounded \emph{symmetric-difference degeneracy (sd-degeneracy)}~\cite{BonnetDSZ24}.
For a graph $G$, it is the least number $c$ such that $G$ can be reduced into a single vertex
by iteratively finding two distinct vertices $u, v$ with 
$|N(v) \triangle N(u)| \leq c$ and removing one of them from the graph.
(This notion of a decomposition is called an \emph{sd-degeneracy sequence}.
See~\cref{sec:prelim} for precise definitions.)
We refer to \cref{fig:diag-classes} for an overview of the properties of graph classes mentioned here.

\begin{figure}[!t]
  \centering
  \begin{minipage}[t]{.48\textwidth}
    \centering
    \includegraphics[width=\linewidth,height=.38\textheight,keepaspectratio]{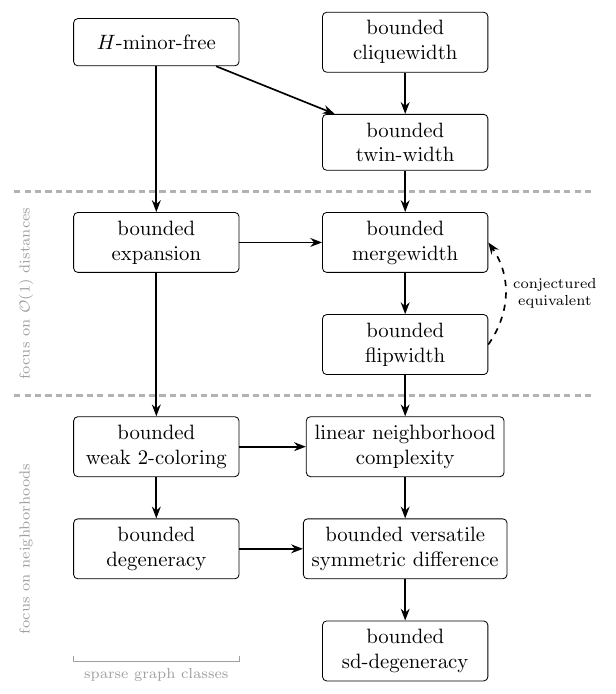}
    \caption{Properties of hereditary graph classes and implications among them.}
    \label{fig:diag-classes}
  \end{minipage}\hfill
  \begin{minipage}[t]{.48\textwidth}
    \centering
    \resizebox{\linewidth}{!}{\begin{tikzpicture}[
    node distance=11mm,
    pipeline/.style={
      draw,
      rounded corners=2pt,
      align=center,
      inner sep=4pt,
      minimum height=8mm,
      font=\small
    },
    implication/.style={-{Stealth[length=2mm]}, thick},
    edge label/.style={font=\small, align=left, text width=60mm}
  ]
  \node[pipeline] (lin) {graph with linear neighborhood complexity};
  \node[pipeline, below=of lin] (sdd) {sd-degeneracy sequence of width $\mathcal{O}(1)$};
  \node[pipeline, below=of sdd] (stm) {signed tree model with $\mathcal{O}(n)$ transversal pairs};
  \node[pipeline, below=of stm] (ibp) {interval biclique partition with $\mathcal{O}(n)$ bicliques};
  \node[pipeline, below=13mm of ibp, xshift=-19mm] (sssp) {SSSP};
  \node[pipeline, below=13mm of ibp, xshift=19mm] (matmult) {matrix multiplication};

  \draw[implication] (lin) -- node[right, edge label]
    {\cref{thm:sdd-seq}, \cref{thm:LV-algo}} (sdd);
  \draw[implication] (sdd) -- node[right, edge label]
    {Lemma 17 of~\cite{BonnetKGM26}, Lemma 3.1 of~\cite{BonnetDSZ24}} (stm);
  \draw[implication] (stm) -- node[right, edge label]
    {Item 1 of Theorem 1 of~\cite{BonnetKGM26}} (ibp);
  \draw[implication] (ibp) -- node[pos=0.55, left, xshift=-1.5mm, font=\small]
    {\cref{thm:sssp-ibp}} (sssp);
  \draw[implication] (ibp) -- node[pos=0.55, right, xshift=1.5mm, font=\small]
    {\cref{lem:matmult-ibp}} (matmult);
\end{tikzpicture}}
    \caption{Diagram of the algorithmic pipeline.}
    \label{fig:diag-pipeline}
  \end{minipage}
\end{figure}

The first algorithmic applications of the linear neighborhood complexity assumption stem from the seminal work of Welzl \cite{Welzl88}, which implies that for a~given $n$-vertex graph from such a class, one can order its vertices in polynomial time so that every neighborhood is a~union of $\Oh(\log n)$ many intervals.
The runtime was improved to $\Oh((m+n)\log n)$ in very recent work with $\Oh(\log^2 n)$ intervals for each vertex~\cite{DreierK26} instead of~$\Oh(\log n)$. 
Prior works \cite{Chan25,Duraj24} achieve $\Oh(\log^2n)$ intervals on average per vertex in the same running time. 
In turn, such a~representation can be transformed into an sd-degeneracy sequence of width $\Oh(\log^2n)$~\cite{BonnetKGM26}.

Furthermore, \cite{BonnetDSZ24} introduced a representation of the edge set of the graph, called a \emph{signed tree model}.
A signed tree model of a graph $G$ is a binary tree with leaves labeled with $V(G)$
and additional transversal edges which can be positive or negative, and which fully determine the adjacency in $G$ (see \Cref{sec:prelim} for a precise definition). 
Subsequent work~\cite{BonnetKGM26} transformed signed tree models into \emph{interval biclique partitions},
first defined in~\cite{twin-width3}.
An interval biclique partition into $b$ bicliques
of a~graph $G$ is a labeling of $V(G)$ with $\{1,2,\ldots,n\}$ together with a partition of $E(G)$
into $b$ bicliques whose sides are intervals of $\{1,2,\ldots,n\}$. 
By pipelining the results of~\cite{BonnetKGM26,BonnetDSZ24}, we can first turn an \mbox{sd-degeneracy}
sequence of width $d$ into a signed tree model with $\Oh(d n)$ transversal pairs in $\Oh(d n+m)$ time,
and then into an interval biclique partition with $\Oh(d n)$ bicliques in $\Oh(d n \log n + m)$ time. 

While sd-degeneracy sequences, sparse signed tree models, and interval biclique partitions are relatively simple and often efficiently computable, they bypass the need for involved decomposition notions in classes of bounded twin-width/merge-width/flip-width, when the problem at hand lends itself to their algorithmic uses.

\subsection{Our results} 

Our first contributions are two time-optimal algorithms for computing sd-degeneracy
sequences of bounded width in classes of linear neighborhood complexity.

\begin{restatable}{theorem}{thmsddseq}\label{thm:sdd-seq}
  For every graph class $\mathcal C$ of linear neighborhood complexity, there is an~$\Oh_{\mathcal C}(n^2)$-time algorithm that, given an $n$-vertex graph $G \in \mathcal C$, outputs an sd-degeneracy sequence of~$G$ of width~$\Oh_{\mathcal C}(1)$.
\end{restatable}
Note that the running time of~\cref{thm:sdd-seq} is optimal for graphs with $\Theta(n^2)$ edges.
This resolves an open problem of~\cite{BonnetKGM26}.
For graphs with a~subquadratic number of edges, we provide a linear-time randomized algorithm.

\begin{restatable}{theorem}{thmlvalgo}\label{thm:LV-algo}
  For every graph class $\mathcal C$ of linear neighborhood complexity, there is a randomized algorithm that, given a~graph $G \in \mathcal C$, outputs an sd-degeneracy sequence of~$G$ of width~$\Oh_{\mathcal C}(1)$ in expected time $\Oh_{\mathcal{C}}\big(|V(G)| + |E(G)|\big)$.
\end{restatable}

Previously available algorithms would either proceed by naively constructing an sd-degeneracy sequence of width $\Oh(1)$, which takes time $\Oh(n^4)$ using a~direct implementation (or $\Oh(n^3)$ after realizing that $\mathcal C$ has bounded versatile symmetric difference, see definition below), or by following Welzl's approach, which currently yields an sd-degeneracy sequence of width $\Oh(\log^2 n)$ in time $\Oh((n+m)\log n)$~\cite{DreierK26,BonnetKGM26}, where $n=|V(G)|$ and $m=|E(G)|$. 

\medskip

To prove \Cref{thm:sdd-seq}, we introduce the notion of \emph{bounded versatile symmetric difference},
where we require that every induced subgraph with $n$ vertices contains $\Theta(n)$ disjoint vertex pairs $u,v$ satisfying $|N(u)\triangle N(v)| \leq c$ for some constant $c$.

On one hand, a simple application of Haussler's packing lemma 
(cf. \cref{thm:lnc-implies-bvsd}) proves that linear neighborhood complexity
implies bounded versatile symmetric difference. 
On the other hand, we show how to use the ``versatility'' strengthening,
together with a~known data structure for longest common extension
of~\cite{LandauV88}, to compute sd-degeneracy sequences
in $\Oh(n^2)$ time.

\medskip

We will actually prove generalizations of~\cref{thm:sdd-seq,thm:LV-algo} to classes of VC density~$\rho$.
\begin{restatable}{theorem}{thmsddseq-gen}\label{thm:sdd-seq-gen}
  For every real number $\rho \geqslant 1$ and every graph class $\mathcal C$ of VC density~$\rho$, there is an~$\Oh_{\mathcal C}(n^{3-\frac{1}{\rho}})$-time algorithm that, given an $n$-vertex graph $G \in \mathcal C$, outputs an sd-degeneracy sequence of~$G$ of width~$\Oh_{\mathcal C}(n^{1-\frac{1}{\rho}})$.
\end{restatable}

\begin{restatable}{theorem}{thmlvalgo-gen}\label{thm:LV-algo-gen}
  For every real number $\rho \geqslant 1$ and every graph class $\mathcal C$ of VC density~$\rho$, there is a randomized algorithm that, given a~graph $G \in \mathcal C$, outputs an sd-degeneracy sequence of~$G$ of width~$\Oh_{\mathcal C}(n^{1-\frac{1}{\rho}})$ in expected time $\Oh_{\mathcal{C}}\big(|V(G)| + |E(G)|\big)$.
\end{restatable}

\medskip
With the aforementioned results of~\cite{BonnetKGM26,BonnetDSZ24},
\cref{thm:sdd-seq} allows us to compute in $\Oh(n^2)$ time a signed tree model with $\Oh(n)$ transversal pairs
and an interval biclique partition with $\Oh(n)$ bicliques. 
Our second contribution is an efficient algorithm using the latter for shortest-path computation.
\begin{restatable}{theorem}{thmsssp}\label{thm:sssp-ibp}
There is an~$\Oh(n+b)$-time algorithm that, given an interval biclique partition of an $n$-vertex graph $G$ with $b$ bicliques and $s \in V(G)$, outputs a~shortest-path tree from vertex~$s$ in~$G$. 
\end{restatable}
Note that the running time of \cref{thm:sssp-ibp} is independent of $m$, the number of edges
of $G$, and in fact is usually sublinear in~$m$.
\Cref{thm:sssp-ibp} bypasses the use of the DAG-compression technique~\cite{Bannach24}, which is a~way to solve shortest-path problems from interval biclique partitions that incurs, in the current state of knowledge, extra superconstant factors.  

By applying \cref{thm:sssp-ibp} to every vertex of the graph and pipelining it with \cref{thm:sdd-seq} and the aforementioned results of~\cite{BonnetKGM26,BonnetDSZ24},
we obtain a time-optimal algorithm for \textsc{All-Pairs Shortest Path} (\sapsp for short) in classes of linear neighborhood complexity.
\begin{restatable}{theorem}{thmlncapsp}\label{thm:lnc-apsp}
  For every graph class $\mathcal C$ of linear neighborhood complexity, \textsc{All-Pairs Shortest Path} can be solved in $\Oh_{\mathcal C}(n^2)$ time on $n$-vertex graphs of~$\mathcal C$.
\end{restatable}
Previous work on \sapsp includes an $\Oh(n^2 \log^2 n)$-time algorithm in classes of bounded twin-width, and in time $\Oh(n^2 \log n)$ on graphs of bounded radius-1 merge-width given with a~witness~\cite{BonnetKGM26}.
More recently, the $\Oh(n^2 \log^2 n)$-time \sapsp algorithm was extended to classes of linear neighborhood complexity~\cite{Cardinal26}.
In~\cite{twin-width3}, an $\Oh(n^2 \log n)$-time \sapsp algorithm was given on graphs of bounded twin-width given with their contraction sequence, which was then improved to $\Oh(n^2)$-time in the paper introducing DAG-compressions~\cite{Bannach24}.

Note that \cref{thm:LV-algo,thm:sssp-ibp}, together with~\cite{BonnetKGM26,BonnetDSZ24}, also give a~randomized $\Oh(kn+m)$-time algorithm for \textsc{Multi-Source Shortest Path} from $k$~sources on $n$-vertex $m$-edge graphs.
Another consequence of~\cref{thm:sdd-seq} and results in~\cite{BonnetKGM26} (see~\cref{lem:matmult-ibp}) is:
\begin{restatable}{theorem}{thmlncmatmult}\label{thm:lnc-matmult}
  For every graph class $\mathcal C$ of linear neighborhood complexity, there is an $\Oh_{\mathcal C}(n^2)$-time algorithm that, given any adjacency matrix $M$ of an $n$-vertex graph of~$\mathcal C$, yields a~data structure that computes the matrix-vector product $MX$ in time $\Oh_{\mathcal C}(n)$ for any $n$-vector $X$.

  In particular, the product $MN$ can be computed in $\Oh_{\mathcal C}(n^2)$ time for \emph{any} $n \times n$ matrix $N$. 
\end{restatable}
We refer to \cref{fig:diag-pipeline} for a diagram of the algorithmic pipeline discussed in the previous few paragraphs.

\cref{thm:lnc-apsp} answers \cite[Question 2]{BonnetKGM26} positively, far beyond classes of bounded twin-width (for which the question was asked). 
\cref{thm:lnc-apsp,thm:lnc-matmult} improve \sapsp and matrix-multiplication algorithms in~\cite{twin-width5,Anand26,BonnetKGM26,Kozma26,Cardinal26}, all of which have extra polylog factors, and extend results in~\cite{KratschN23,Bannach24}, which work on graphs of bounded clique-width and on graphs of bounded twin-width given with contraction sequences, respectively.

By~using \cref{thm:sdd-seq-gen} rather than \cref{thm:sdd-seq}, one can solve \textsc{All-Pairs Shortest Path} in graphs of VC density~$\rho$ (or multiply the adjacency matrix of graphs of VC density~$\rho$ with any matrix) in time $\Oh(n^{3-1/\rho})$.
However, we observe that fast matrix multiplication and Seidel's algorithm~\cite{Seidel95} have a~better running time in general instances as soon as $\rho > 1.591$.

\medskip
We also obtain a~time-optimal triangle-detection algorithm in classes of linear neighborhood complexity, which leverages an sd-degeneracy sequence of constant width.

\begin{restatable}{theorem}{triangle-detection}\label{thm:triangle-detection}
For every graph class $\mathcal C$ of linear neighborhood complexity, \textsc{Triangle Detection} can be solved in expected time $\Oh_{\mathcal C}(n+m)$ on $n$-vertex $m$-edge graphs of $\mathcal C$.
\end{restatable}

Let us recall that on general graphs, the fastest known \textsc{Triangle Detection} algorithm runs in time $\Oh(\min(n^\omega,m^{2\omega/(\omega+1)}))$~\cite{Alon97} where $\omega$ is the exponent of square matrix multiplication, which currently implies $\Oh(\min(n^{2.372},m^{1.407}))$ time.

Utilizing both an sd-degeneracy sequence and a~Welzl order of the input graph, we obtain randomized almost-linear algorithms for finding 4-vertex and 5-vertex cliques, and a~deterministic $\Oh(n^2)$-time algorithm for $K_4$ detection (time-optimal within dense graphs).    

\begin{restatable}{theorem}{k4-detection}\label{thm:k4-detection}
For every graph class $\mathcal C$ of linear neighborhood complexity, a~$K_4$ subgraph can be found in expected $\Oh_{\mathcal C}(n \log^5 n + m \log n)$ or in deterministic $\Oh_{\mathcal C}(n^2)$ time on $n$-vertex $m$-edge graphs of~$\mathcal C$.
\end{restatable}

\begin{restatable}{theorem}{k5-detection}\label{thm:k5-detection}
For every graph class $\mathcal C$ of linear neighborhood complexity, a~$K_5$ subgraph can be found in expected time $\Oh_{\mathcal C}(n \log^9 n + m \log^5 n)$ on $n$-vertex $m$-edge graphs of $\mathcal C$.
\end{restatable}

Prior to our work, no linear-time algorithm was known for \textsc{Triangle Detection} in the much more restricted classes of bounded clique-width, nor was an almost-linear algorithm known in classes of bounded twin-width (where, for both, the input is a~mere graph without a~witness of~low width).

\section{Preliminaries}\label{sec:prelim}

We index the word positions from 1.
In particular, $w = w[1]w[2] \cdots w[s]$ for every word $w$ of length~$s$.
We let $\log x$ denote the base-$2$ logarithm of $x$.
For integers $\ell \leq r$, let $[\ell, r]$ be the integer interval spanning all values from $\ell$ to $r$, inclusively, and for integer $n$, let $[n] = [1, n]$.
Graphs are simple (without loops or parallel edges), undirected, and finite.
For a graph $G$, we denote its vertex- and edge-set as $V(G)$ and $E(G)$, respectively.

\subsection{Neighborhood complexity}
For a graph $G$, let $\pi_G(m)$ be the maximum number of distinct neighborhoods of vertices in~$G$ restricted to a~set of at~most $m$ vertices, i.e.,
\[\pi_G(m) := \max_{X \subseteq V(G), |X| \le m} \abs{\{N_G(v) \cap X : v \in V(G)\}}.\]
We say that $G$ has \emph{$c$-linear neighborhood complexity} (or $c$-LNC) if $\pi_G(m) \le c\cdot m$ for all positive $m \in \Nn$.
Class $\Cc$ has \emph{linear neighborhood complexity} if there exists a~constant~$c$ such that all graphs $G \in \Cc$ have $c$-LNC.
We say that $\Cc$ has \emph{VC density $\rho$} if there exists a~constant~$c$ such that for all $G \in \Cc$, $\pi_G(m) \leqslant cm^\rho$ for all positive $m$.
Finally, we say that $G$ has \emph{$(c,\rho)$-polynomial neighborhood complexity} (or $(c,\rho)$-PNC) if $\pi_G(m) \le c\cdot m^\rho$ for all positive $m \in \Nn$.

\subsection{Symmetric difference and sd-degeneracy} 

In this paper, we define the \emph{symmetric difference} of two vertices $u, v$ in a~graph $G$ as
\[\sd_G(u,v) := \abs{N_G(u)\,\triangle\,N_G(v)},\]
where $\triangle$ denotes the symmetric difference of two sets: $X \triangle Y := (X \setminus Y) \cup (Y \setminus X)$.
There are slight variations in defining the symmetric difference of $u, v$ (depending on whether or not $u, v$ themselves count), but they only differ by a~constant additive term.
We say that two distinct vertices $u, v$ are \emph{$d$-near-twins} in~$G$ if $\sd_G(u,v) \leqslant d$,
and are \emph{twins} if $\sd_G(u,v) = 0$.

The \emph{symmetric difference} $\sd(G)$ of a~graph $G$ is the least nonnegative integer~$d$ such that for every induced subgraph $H$ of~$G$ with at least two vertices, there are two $d$-near-twins in~$H$.

An~\emph{sd-degeneracy sequence} of an $n$-vertex graph $G$ is a~list $(u_1, v_1), \ldots, (u_{n-1},v_{n-1})$ of pairs of vertices of $G$ such that
\medskip
\begin{compactitem}
\item $u_i \neq v_i$ for every $i \in [n-1]$,
\item $V(G) = \{u_1, u_2, \ldots, u_{n-2},$ $u_{n-1}, v_{n-1}\}$, and
\item there is no $i < j$ such that $u_i \in \{u_j, v_j\}$.
\end{compactitem}
\medskip
Informally, it is an elimination sequence where only the first element of the pair is removed at each step. 
The \emph{width} of the sd-degeneracy sequence $(u_1, v_1), \ldots, (u_{n-1},v_{n-1})$ is defined as \[\max_{i \in [n-1]} \sd_{G - \{u_1, u_2, \ldots, u_{i-1}\}}(u_i,v_i).\]
The \emph{sd-degeneracy} of~$G$ (where 'sd' stands for symmetric difference) is the least integer~$d$ such that $G$ admits an sd-degeneracy sequence of width~$d$.

Equivalently, the class of graphs of sd-degeneracy at most $d$ can be defined by induction on the number of vertices: a graph $G$ has sd-degeneracy at most $d$ if either $|V(G)| \leq 1$, or there is a pair of distinct vertices $u, v \in V(G)$ such that $\sd_G(u,v) \leq d$ and the graph $G - u$ has sd-degeneracy at most $d$.

We say that a~hereditary graph class $\mathcal C$ has \emph{bounded versatile symmetric difference} if there is an integer $d := d_{\mathcal C}>0$ such that for every $n$-vertex graph $G \in \mathcal C$, there are at~least $\lfloor n/d \rfloor$ disjoint pairs of vertices $u, v$ satisfying $\sd_G(u,v) \leqslant d$.
For $\alpha \in [0,1)$, we say that $\mathcal C$ has \emph{versatile symmetric difference of exponent~$\alpha$} if there is a~constant~$d:= d_{\mathcal C} > 0$ such that for every $n$-vertex graph $G \in \mathcal C$, there are $\lfloor n/d \rfloor$ disjoint pairs of vertices $u, v$ in~$G$ satisfying $\sd_G(u,v) \leqslant d n^\alpha$.
Thus \emph{versatile symmetric difference of exponent~0} coincides with \emph{bounded versatile symmetric difference}.

\subsection{Signed tree models and interval biclique partitions}

We recall the definition of signed tree models for the sake of self-containedness.
However, we will not use them in the rest of the paper.
Thus, from this subsection, one can only read the definition of interval biclique partitions (given just after~\cref{lem:sdd-seq-to-stm}) and \cref{lem:sdd-seq-to-ibp}.

\emph{Signed tree models}, originally defined in~\cite{BonnetDSZ24}, generalize the tree models introduced in the context of twin-width~\cite{BonnetNMST24,twin-width3}.
In a rooted tree $T$ and $u,v\in V(T)$, by $u \preceq_T v$ we denote that $u$ is an ancestor of $v$.
A~pair of vertices of a~rooted tree $T$ is a~\emph{transversal pair} of~$T$ if it is not between an ancestor and descendant.
We say that two transversal pairs $u_1v_1$ and $u_2v_2$ \emph{cross} if one of $u_1, v_1$ is a~strict ancestor of one of $u_2, v_2$, and conversely, one of $u_2, v_2$ is a~strict ancestor of one of $u_1, v_1$.

A~\emph{full} binary tree is a~rooted binary tree where every internal node has exactly two children. 
A~\emph{signed tree model} is a~triple $\mathcal T=(T, A, B)$ where $T$ is a~full binary tree, $A$ and $B$ are two disjoint sets of transversal pairs of $T$, and $A \cup B$ does not contain any crossing pair.
The transversal pairs of $A$ and $B$ are called \emph{negative edges} and \emph{positive edges}, respectively.

We say that a~transversal pair $u'v'$ \emph{covers} the pair $u,v \in V(T)$ if $u'\preceq_T u$ and $v'\preceq_T v$ and there is no transversal pair $u''v'' \neq u'v'$ in $A\cup B$ with $u'\preceq_T u''\preceq_T u$ and $v'\preceq_T v'' \preceq_T v$.
We denote by $L(T)$ the set of leaves of~$T$.
The graph $G_{\mathcal T}$ represented by the signed tree model $\mathcal T$ has vertex set $L(T)$, and $u,v \in L(T)$ are adjacent in~$G_{\mathcal T}$ if and only if the pair $u,v$ is covered by a~positive transversal pair in~$\mathcal T$.

\begin{figure}[!ht]
  \centering
  \begin{tikzpicture}[%
    scale=1.3,
    level distance=10mm,
    level 1/.style={sibling distance=32mm},
    level 2/.style={sibling distance=16mm},
    level 3/.style={sibling distance=8mm},
    level 4/.style={sibling distance=4mm},
    every node/.style={draw, circle, inner sep=2pt, font=\footnotesize},
    edge from parent path={(\tikzparentnode) -- (\tikzchildnode)}
  ]
  \node (root) {}
    child {node (a) {}
      child {node (b) {}
        child {node (c) {}
          child {node (c1) [label=below:{1}] {}}
          child {node (c2) [label=below:{2}] {}}
        }
        child {node (d) [label=below:{3}] {}}
      }
      child {node (e) {}
        child {node (f) {}
          child {node (f1) [label=below:{4}] {}}
          child {node (f2) [label=below:{5}] {}}
        }
        child {node (g) {}
          child {node (g1) [label=below:{6}] {}}
          child {node (g2) [label=below:{7}] {}}
        }
      }
    }
    child {node (h) {}
      child {node (i) {}
        child {node (j) {}
          child {node (j1) [label=below:{8}] {}}
          child {node (j2) [label=below:{9}] {}}
        }
        child {node (k) [label=below:{10}] {}}
      }
      child {node (l) {}
        child {node (m) {}
          child {node (m1) [label=below:{11}] {}}
          child {node (m2) [label=below:{12}] {}}
        }
        child {node (n) {}
          child {node (n1) [label=below:{13}] {}}
          child {node (n2) [label=below:{14}] {}}
        }
      }
    };

    \foreach \i/\j/\b in {a/n/15, c2/k/0, c2/d/0}{
      \draw[very thick, amber] (\i) to [bend left=\b] (\j);
    }

    \draw[very thick, amber] (i) to node[midway, above, inner sep=1pt, draw=none] {$e_2$} (e);
    
    \foreach \i/\j/\b in {e/k/0, g1/g2/0, m1/m2/0, m2/n/0, f1/g/0, n/i/0, c/d/0}{
       \draw[very thick, blue] (\i) to [bend left=\b] (\j);
     }

\draw[very thick, blue] (f) to [bend left=28] node[midway, above, inner sep=1pt, draw=none] {$e_1$} (j);
\draw[very thick, blue] (a) to [bend left=0] node[midway, above, inner sep=1pt, draw=none] {$e_3$} (h);

    \draw[very thin, dashed, red] (b) -- (h) ; 
\end{tikzpicture}
  \caption{A~signed tree model of a~14-vertex graph, with~$A$ in amber and~$B$ in blue.
    The topmost amber edge and the dashed red edge cross (which is not in $A \cup B$).
    Vertex 8 is adjacent to 4 because 4,8 is covered by the positive edge $e_1$, and to 2 because 2,8 is covered by the positive edge $e_3$, but 8 is not adjacent to 7 because 7,8 is covered by the negative edge $e_2$.}
\label{fig:signed-tree-model}
\end{figure}
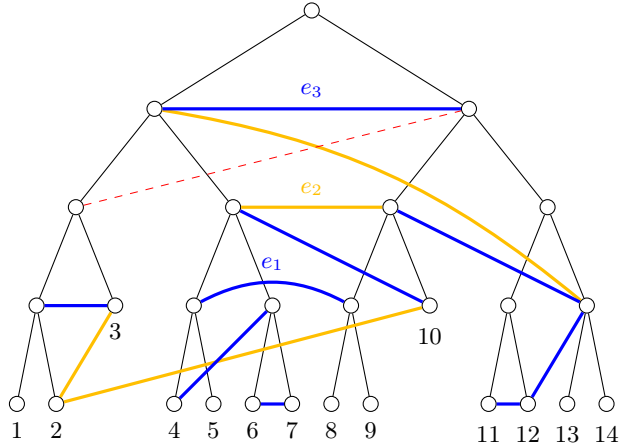

A~graph class $\mathcal C$ has \emph{sparse signed tree models} if there is a~constant $c$ such that every $G \in \mathcal C$ admits a~signed tree model $(T,A,B)$ with $|A \cup B| \leqslant c \cdot \abs{V(T)}$.
There is an efficient algorithm that turns sd-degeneracy sequences of bounded width into sparse signed tree models.  

\begin{lemma}[Lemma 18 of~\cite{BonnetKGM26}, Lemma 3.1 of~\cite{BonnetDSZ24}]\label{lem:sdd-seq-to-stm}
  There is an~$\Oh(dn+m)$-time algorithm that, given an sd-degeneracy sequence of an $n$-vertex $m$-edge graph $G$ of width~$d$, outputs a~signed tree model of~$G$ with $\Oh(dn)$ transversal pairs. 
\end{lemma}

Assume that the graph $G$ comes with a~total ordering of its vertices; for convenience, assume that the vertex set of $G$ is $[n]$, ordered naturally.
We then define the \emph{interval biclique partition} of $G$ as the partitioning of the edge set of $G$ into bicliques, where each side of each biclique is an~integer subinterval of $[n]$.
The following lemma implies that a~small-size signed tree model can be efficiently turned into a~small-size interval biclique partition.

\begin{lemma}[Item 1 of Theorem 1 of~\cite{BonnetKGM26}]\label{lem:stm-to-ibp}
  There is an~$\Oh(p \log n)$-time algorithm that, given a~signed tree model of an $n$-vertex graph $G$ with $p$ transversal pairs, outputs an interval biclique partition of~$G$ with $\Oh(p)$ bicliques.
\end{lemma}

We will rely on the following consequence of \cref{lem:sdd-seq-to-stm,lem:stm-to-ibp}.

\begin{lemma}\label{lem:sdd-seq-to-ibp}
   There is an~$\Oh(d n \log n + m)$-time algorithm that, given an sd-degeneracy sequence of an $n$-vertex $m$-edge graph $G$ of width~$d$, outputs an interval biclique partition of~$G$ with $\Oh(d n)$ bicliques.
\end{lemma}

\subsection{Matrix multiplication}

An interval biclique partition with $\Oh(n)$ bicliques of an $n$-vertex graph $G$ allows one to multiply any adjacency matrix of~$G$ with any $n$-vector in $\Oh(n)$ time.
More precisely:

\begin{lemma}[Theorem 33 in~\cite{BonnetKGM26}]\label{lem:matmult-ibp}
  For every abelian group $\mathcal A$, there exists an algorithm
  that takes on input an interval biclique partition with $b$ bicliques
  of an $n$-vertex graph $G$ with adjacency matrix $M$
  and an $n$-vector $X \in \mathcal A^n$ and computes the product $MX$,
  in time $\Oh(n+b)$ and $\Oh(n+b)$ group operations. 
\end{lemma}
Note that, as the adjacency matrix $M$ is a $\{0,1\}$-matrix, the multiplication
of $M$ with a vector of group elements is well-defined. 
This lemma has a~short proof that encodes $M$ into an $n \times n$ $\{-1,0,1\}$-matrix with $\Oh(b)$ nonzero entries.
As a~direct consequence of~\cref{lem:matmult-ibp}, one can multiply $M$ with any $n \times n$ matrix in time $\Oh(n(n+b))$
(assuming the group operations take constant time).

\section{Finding an sd-degeneracy sequence deterministically}\label{sec:sd-deg-det}

In this section, we provide an algorithm that computes an sd-degeneracy sequence of constant width in graphs of bounded versatile symmetric difference.
The algorithm is time-optimal for dense graphs (i.e., $n$-vertex graphs with $\Theta(n^2)$ edges).
More generally, the algorithm computes an sd-degeneracy sequence of width $\Oh(n^\alpha)$ in graphs of versatile symmetric difference of exponent~$\alpha$.

It uses classic data structures for strings.
Given a~word $w$ of finite length~$s$, a~longest common extension (LCE) query $i,j \in [s]$ asks for the largest integer~$k$ such that $w[i]w[i+1] \cdots w[i+k-1] = w[j]w[j+1] \cdots w[j+k-1]$.  

\begin{theorem}[Section 4 of~\cite{LandauV88}]\label{thm:lce}
  Fix a~finite alphabet $\Sigma$.
  There is an~$\Oh(s)$-time algorithm, that given a~word $w$ of length~$s$ over~$\Sigma$, builds a~data structure that answers each LCE query on~$w$ in constant time.
\end{theorem}

\Cref{thm:lce} is originally shown by reduction to least common ancestor (LCA) in the suffix tree.
To avoid computing the suffix tree, one can alternatively go through suffix array, longest common prefix (LCP) array, and range minimum query (RMQ).  
See~\cite{BenderF00} for a~simple proof of optimal LCA and RMQ data structures.

\Cref{thm:lce} yields a~surprisingly fast check for $d$-near-twinness in graphs.

\begin{theorem}\label{thm:near-twin-ds}
  There is an~$\Oh(n^2)$-time algorithm that, given an~$n$-vertex graph~$G$, builds a~data structure that given a~query $u, v \in V(G), d \in \mathbb N$ reports if $\sd_G(u,v) \leqslant d$ in $\Oh(d)$ time.  
\end{theorem}

\begin{proof}
  We build a~binary word $w$ of length~$n^2$ by concatenating the adjacency vectors of all vertices.
  Formally, assuming $V(G)=[n]$, for every $a, b \in [n]$, $w[(a-1)n+b]=1$ if $ab \in E(G)$ and $w[(a-1)n+b]=0$ otherwise.
  The word $w$ can easily be built in $\Oh(n^2)$ time.

  By~\cref{thm:lce}, we build, in $\Oh(n^2)$ time, a~data structure $\mathcal D_w$ that answers each LCE query on~$w$ in constant time.
  Given a~query $u, v \in V(G), d \in \mathbb N$, we proceed as follows.
  We initialize $p \leftarrow 1$ and $q \leftarrow d$.

  While $p \leqslant n$ and $q \geqslant 0$, we feed $\mathcal D_w$ with the query $(u-1)n+p, (v-1)n+p$.
  (Note that, for $p=1$, those are the positions where the adjacency vectors of $u$ and $v$, respectively, start.)
  Let $k$ be the query answer, found in constant time.
  If $p+k \leqslant n$, we decrement $q$ by 1 and set $p \leftarrow p+k+1$, else we exit the while loop.

  When we exit the while loop, if $q \geqslant 0$, we report that $\sd_G(u,v) \leqslant d$, and otherwise ($q=-1$) that $\sd_G(u,v) > d$.

  \medskip

  \textbf{Time.}
  We have already checked that $\mathcal D_w$ is built in $\Oh(n^2)$ time.
  Note that each query $u, v, d$ makes at most $d+1$ LCE queries.
  Indeed, in the end of the while loop, we decrement $q$ (initially set to~$d$) by 1 or we exit the loop, whereas continuing looping requires that $q$ is nonnegative.
  Hence a~query is met in time $(d+1) \Oh(1)+\Oh(1)=\Oh(d)$.

  \medskip

  \textbf{Correctness.}
  Let $1, p_1, \ldots, p_h$ be the successive values $p$ is set to.
  In particular, $h \leqslant d$.
  The following properties hold:
  \begin{itemize}
  \item $1 < p_1 < \ldots < p_h \leqslant n+1$;
  \item for every $i \in [h]$, $w[(u-1)n + p_i - 1] \neq w[(v-1)n + p_i - 1]$;
  \item for every $j \in [p_h-1]$, if $w[(u-1)n + j] \neq w[(v-1)n + j]$ then there exists $i \in [h]$ such that $j=p_i-1$.
  \end{itemize}
  The main invariant is that there are exactly $i$ vertices with label smaller than $p_i$ adjacent to exactly one of~$u, v$.
  As we check both if more than $d$ such positions where found ($q \leqslant 0$) and if the current position overflows to the next adjacency vectors ($p \leqslant n$), we correctly decide whether or not $\sd_G(u,v) \leqslant d$ holds when exiting the while loop.
\end{proof}

In particular, we can list \emph{all} pairs of $d$-near-twins in quadratic time, when $d$ is constant.

\begin{corollary}\label{cor:near-twin-listing}
  There is an~$\Oh(d n^2)$-time algorithm that, given an~$n$-vertex graph and a~positive integer $d$, lists all the pairs of vertices $u, v \in V(G)$ such that $\sd_G(u,v) \leqslant d$.  
\end{corollary}

\begin{proof}
  We compute in $\Oh(n^2)$ time the data structure $\mathcal D$ of~\cref{thm:near-twin-ds}.
  For every $u, v \in V(G)$ with $u \neq v$, we present $\mathcal D$ with the query $u, v, d$.
  If the answer is positive ($\sd_G(u,v) \leqslant d$), we add $u,v$ to the output pairs.
  Overall, this takes $\Oh(d n^2)$ time.
\end{proof}

We obtain the objective of this section by repeatedly appling~\cref{cor:near-twin-listing}.

\begin{theorem}\label{thm:sd-degen-seq}
  Let $\mathcal C$ be a~(hereditary) class of versatile symmetric difference of~exponent~$\alpha$.
  There is an~$\Oh_{\mathcal C}(n^{2+\alpha})$-time algorithm that, given an~$n$-vertex graph of~$\mathcal C$, outputs an sd-degeneracy sequence of~$G$ of width~$\Oh_{\mathcal C}(n^\alpha)$.  
\end{theorem}

\begin{proof}
  Let $d := d_{\mathcal C}>0$ be such that for every $n$-vertex graph $G \in \mathcal C$, there are $\lfloor n/d \rfloor$ disjoint pairs of vertices $u, v$ satisfying $\sd_G(u,v) \leqslant d n^\alpha$.
  We initialize an sd-degeneracy sequence $\mathcal S$ of~$G$ of width~$d n^\alpha$ to the empty list.
  We also set $G_1 := G$ and $n_1 := n$.

  While the remaining graph $G_i$ has $n_i \geqslant d$ vertices, we do the following.
  
  By~\cref{cor:near-twin-listing}, we list all the pairs of $d n^\alpha$-near-twins of $G_i$ in $\Oh(d n^\alpha n_i^2)$ time.
  We then greedily build a~list $\{u_1,v_1\}, \ldots, \{u_h,v_h\}$ of $\max(1,\lfloor \frac{1}{2} \lfloor n_i/d \rfloor \rfloor)$ \emph{disjoint} pairs of $d n^\alpha$-near-twins of $G_i$ in $\Oh(n_i)$ time.
  (Indeed, if one fixes a~set $D$ disjoint pairs, then the choice of any other pair consumes at most two pairs of~$D$.)
  We append this list to $\mathcal S$.
  We then set $G_{i+1} := G_i - \{u_1, \ldots, u_h\}$ and $n_{i+1} := |V(G_{i+1})|$.
  This finishes the body of the while loop.

  When we exit the while loop, we finish the sd-degeneracy sequence arbitrarily.

  \medskip

  \textbf{Correctness.}
  If the algorithm terminates, it does output an sd-degeneracy sequence of $G$ of width~$d n^\alpha$ since we only append to~$\mathcal S$ pairs of $d n^\alpha$-near-twins in the current graph, and remove their first elements in batches of disjoint pairs.

  \medskip

  \textbf{Time.}
  At every loop iteration, we add at~least one pair to the sequence and remove their first elements from the graph.
  Thus, the algorithm does terminate, and goes through the successive induced subgraphs of~$G$: $G_1, G_2, \ldots, G_p$ with $p \leqslant n$.
  Let $c := 1-\frac{1}{4d}$.

  \begin{claim}\label{clm:vertex-count-decr}
    For every $i \in [p-1]$, $n_{i+1} \leqslant c n_i$.
  \end{claim}
  \begin{proof}
    Indeed, if $n_i < 4d$, then $n_{i+1} \leqslant c n_i$ since $n_{i+1} \leqslant n_i - 1$ (we remove at least one vertex of $G_i$ to form $G_{i+1}$).
    If instead, $n_i \geqslant 4d$, we conclude because  $\max(1,\lfloor \frac{1}{2} \lfloor n_i/d \rfloor \rfloor)=\lfloor \frac{1}{2} \lfloor n_i/d \rfloor \rfloor \geqslant \frac{1}{4d} \cdot n_i$.
  \end{proof}

  By~\cref{clm:vertex-count-decr}, for every $i \in [p]$, $n_i \leqslant c^{i-1} n$.
  So the overall running time is
  \[\sum\limits_{i \in [p]} \Oh(d n^\alpha n^2_i) \leqslant \Oh(d n^\alpha) \cdot n^2 \sum\limits_{i \in [p]} c^{2i-2} \leqslant \frac{1}{1-c} \cdot \Oh_{\mathcal C}(n^{2+\alpha}) = \Oh_{\mathcal C}(n^{2+\alpha}),\]
  and we conclude.
\end{proof}

\Cref{thm:sdd-seq-gen} is then a~consequence of~\cref{thm:sd-degen-seq,thm:lnc-implies-bvsd}.
For classes of bounded versatile symmetric difference (exponent~0), we get the following corollary (which, by~\cref{thm:lnc-implies-bvsd}, implies~\cref{thm:sdd-seq}).

\begin{corollary}\label{thm:sd-degen-seq-cor}
  Let $\mathcal C$ be a~(hereditary) class of bounded versatile symmetric difference.
  There is an~$\Oh_{\mathcal C}(n^2)$-time algorithm that, given an~$n$-vertex graph of~$\mathcal C$, outputs an sd-degeneracy sequence of~$G$ of width~$\Oh_{\mathcal C}(1)$.  
\end{corollary}

\section{SSSP in graphs given with an interval biclique partition}

In this section we prove \cref{thm:sssp-ibp}, which for convenience we recall below.

\thmsssp*
\newcommand{\IntvDS}{\textsc{Interval-Intersect}\xspace}

We begin by introducing a~data structure problem for reporting interval intersections, called \IntvDS, and we show how to solve this problem in total linear time.
In \IntvDS, we are given as input a~unary-encoded integer $N \geq 1$ (the universe size) and a~finite set $A$ of labels, where every label $a \in A$ is assigned an~integer interval $I_a \subseteq [N]$.
Initially, all labels are active.
The data structure is required to answer the following online queries: given an~integer interval $J \subseteq [N]$, report all active labels $a \in A$ with the corresponding intervals $I_a$ intersecting $J$, and deactivate all reported labels.

Subsequently, we will use the resulting data structure to show the linear-time algorithm for \Cref{thm:sssp-ibp}.
In \cref{ss:bootstrap}, we prove the following.
\begin{lemma}
  \label{lem:intv-ds-linear}
  There exists a~data structure for \IntvDS in the word RAM model of computation that processes an~instance with universe size $N$, set of labels $A$, and $q$~queries in $\Oh(N + |A| + q)$ time.
\end{lemma}
The purpose of encoding $N$ in unary is to ensure that every element of the universe $[N]$ fits within a~single RAM word; hence we can assume $N < 2^b$ where $b$ is the size of the machine word.

We show now to use \Cref{lem:intv-ds-linear} to solve SSSP in linear time in graphs with a~given interval biclique partition, namely, we prove \Cref{thm:sssp-ibp}.
\begin{proof}[Proof of \Cref{thm:sssp-ibp}]
  Recall that $[n]$ is the vertex set of the graph, and let the $i$th biclique in the edge partition ($i \in [b]$) have sides $A_i$ and $B_i$; each $A_i$ and $B_i$ is a~subinterval of $[n]$.

  In the beginning, we set up two instances of the linear-time data structure for \IntvDS, both with universe $[n]$:
  \newcommand{\Dgather}{D_{\mathrm{gather}}}
  \newcommand{\Dscatter}{D_{\mathrm{scatter}}}
  \begin{itemize}
    \item $\Dgather$, with labels from $[n]$, where the label $i \in [n]$ is assigned the interval $[i, i]$;
    \item $\Dscatter$, with labels $\bigcup_{i \in [b]} \{a_i, b_i\}$, with the interval $A_i$ assigned to $a_i$ and the interval $B_i$ assigned to $b_i$.
  \end{itemize}

  The active items in $\Dgather$ are the vertices undiscovered by the run of the breadth-first search (BFS); hence, we initially query $\Dgather$ with $[s, s]$, where $s$ is the source of the BFS, to deactivate the label $s$.
  On the other hand, the active items in $\Dscatter$ represent the sides of the bicliques that do not contain vertices visited by BFS.

  We perform the usual BFS in $G$, with the usual queue holding the vertices discovered by the search (initially only $s$) that are yet to be visited.
  When a~vertex $v \in [n]$ is visited by the search, we find the yet undiscovered neighbors of $v$ as follows.
  We first use $\Dscatter$ to enumerate the not yet processed sides of the bicliques by querying $\Dscatter$ for the interval $[v, v]$.
  Suppose $\Dscatter$ returns $a_i$ (resp., $b_i$) as one of the labels; then, we have that $v \in A_i$ (resp., $v \in B_i$).
  In this case, all vertices in $B_i$ (resp., $A_i$) should become discovered if not yet discovered.
  In order to ensure that, we query $\Dgather$ with the interval $B_i$ (resp., $A_i$), and enqueue all vertices returned as a~result of the query, marking $v$ as the parent of each of them.
  We repeat this process for each label reported by $\Dscatter$.

  The correctness of the algorithm is clear as it is simply a~breadth-first search with a~bespoke method for listing the yet undiscovered neighbors.
  To see the efficiency, observe that we issue at most $n$ queries to $\Dscatter$, and it reports at most $\Oh(b)$ labels throughout the entire run of SSSP.
  Thus, at most $\Oh(b)$ queries to $\Dgather$ are issued.
  Hence both data structures for \IntvDS process all respective queries in total $\Oh(n + b)$ time.
  Moreover, it is clear that outside of the queries, the SSSP implementation above takes time linear in the size of the input; the claim follows.
\end{proof}

\subsection{Interval intersections data structure}\label{ss:bootstrap}

It remains to prove \cref{lem:intv-ds-linear}.

Our approach resembles the \emph{tabulation technique} and the \emph{method of Four Russians}, commonly used in algorithm design in the word RAM model of computation; see e.g.\ \cite{ArlazarovDinicEtAl1970,GabowTarjan1985,BerkmanVishkin1994,BenderF00}.
Namely, we first show a~data structure with a worse complexity guarantee $\Oh(N \log N + |A| + q)$ (\Cref{lem:intv-ds-quasilinear}).
Using \Cref{lem:intv-ds-quasilinear}, we then show how to bootstrap a~linear-time data structure for universes of size $n \leq N$, given access to a~linear-time data structure for universes of size $\log n$ (\Cref{lem:intv-ds-bootstrapping}).
Next, in the word RAM model, we give a~linear-time data structure for universes of size $\sqrt{\log N}$, assuming $\Oh(N)$-time prior preprocessing (\Cref{lem:intv-ds-tiny-universe}).
Applying \Cref{lem:intv-ds-bootstrapping} twice to \Cref{lem:intv-ds-tiny-universe} implies \Cref{lem:intv-ds-linear}.
We implement this strategy below.

\begin{lemma}[Quasi-linear interval intersections]
  \label{lem:intv-ds-quasilinear}
  There exists a~data structure for \IntvDS that processes an~instance with universe size $N$, set of labels $A$, and $q$~queries in $\Oh(N \log N + |A| + q)$ time.
\end{lemma}
\begin{proof}
  On initialization, we build an~auxiliary directed graph $H$ with $\Oh(N \log N + |A|)$ vertices and arcs, as follows.
  Let $L = \left\lfloor \log N \right\rfloor$.
  For integer $k \in [0, L]$, let the \emph{$k$th interval layer} $\mathcal{I}_k$ comprise all integer intervals $[i, i + 2^k - 1] \subseteq [N]$.
  For every such $k$ and every interval $I \in \mathcal{I}_k$, we introduce two vertices: an~outvertex $u_I$ and an~invertex $v_I$.
  For every $a \in A$, we also introduce a~vertex $a$.

  We construct arcs in $H$ as follows.
  Every $I \in \mathcal{I}_k$ with $k \geq 1$ is a~disjoint union of two integer intervals $I_1, I_2 \in \mathcal{I}_{k-1}$; add arcs $u_I \to u_{I_1}$, $u_I \to u_{I_2}$ (so that arcs between outvertices are directed toward lower-indexed layers) and arcs $v_{I_1} \to v_I$, $v_{I_2} \to v_I$ (so arcs between invertices are oriented toward higher-indexed layers).
  Also, for every $I = [i, i] \in \mathcal{I}_0$, add an~arc $u_I \to v_I$.
  It is easy to verify that now, for every $I \in \mathcal{I}_k$ and $x \in [N]$, there is a~path from $u_I$ to $u_{[x, x]}$ if and only if $x \in I$; and by the same token, there is a~path from $v_{[x, x]}$ to $v_I$ if and only if $x \in I$.
  It follows that for an~outvertex $u_J$ and an~invertex $v_I$, a~path from $u_J$ to $v_I$ exists if and only if the intervals $J$ and $I$ intersect.
  Finally, for every $a \in A$ with the corresponding interval $I_a$, write $I_a$ as a~(not necessarily disjoint) union of two intervals $I_1, I_2 \in \mathcal{I}_k$ for some $k \in [0, L]$ (where possibly $I_1 = I_2$).
  Then add arcs $v_{I_1} \to a$, $v_{I_2} \to a$.
  This guarantees that for any outvertex $u_J$, a~path from $u_J$ to $a$ in $H$ exists if and only if $J$ and $I_a$ intersect.
  See~\cref{fig:interval-ds}.

\begin{figure}[!ht]
\centering
\resizebox{250pt}{!}{
\begin{tikzpicture}[
    vertex/.style={
        circle,
        draw,
        inner sep=0pt,
        minimum size=3.2pt
    },
    marked/.style={
        circle,
        draw,
        very thick,
        fill=white,
        inner sep=0pt,
        minimum size=4.8pt
    },
    labelvertex/.style={
        circle,
        draw,
        inner sep=1.4pt,
        font=\scriptsize
    },
    arc/.style={
        -stealth,
        gray!75,
        line width=.35pt,
        shorten >=1pt,
        shorten <=1pt
    },
    patharc/.style={
        -stealth,
        very thick,
        shorten >=1pt,
        shorten <=1pt
    },
    note/.style={
        font=\scriptsize,
        inner sep=1pt
    }
]

\def\N{8}
\def\L{3}
\def\dy{1.15}
\def\shift{.55}
\def\ax{10.2}

\foreach \k in {0,...,\L} {
    \pgfmathtruncatemacro{\len}{2^\k}
    \pgfmathtruncatemacro{\last}{\N-\len+1}
    \pgfmathsetmacro{\yu}{\dy*(\k+\shift)}
    \pgfmathsetmacro{\yv}{-\dy*(\k+\shift)}

    \node[note, anchor=east] at (-.45,\yu) {$\mathcal I_{\k}$};
    \node[note, anchor=east] at (-.45,\yv) {$\mathcal I_{\k}$};

    \foreach \i in {1,...,\last} {
        \pgfmathsetmacro{\x}{\i+\len/2-1}
        \node[vertex] (u-\k-\i) at (\x,\yu) {};
        \node[vertex] (v-\k-\i) at (\x,\yv) {};
    }
}

\node[note] at (3.9,4.8) {outvertices $u_I$};
\node[note] at (3.9,-4.8) {invertices $v_I$};

\foreach \k in {1,...,\L} {
    \pgfmathtruncatemacro{\len}{2^\k}
    \pgfmathtruncatemacro{\half}{2^(\k-1)}
    \pgfmathtruncatemacro{\last}{\N-\len+1}
    \pgfmathtruncatemacro{\prev}{\k-1}

    \foreach \i in {1,...,\last} {
        \pgfmathtruncatemacro{\j}{\i+\half}

        \draw[arc] (u-\k-\i) -- (u-\prev-\i);
        \draw[arc] (u-\k-\i) -- (u-\prev-\j);

        \draw[arc] (v-\prev-\i) -- (v-\k-\i);
        \draw[arc] (v-\prev-\j) -- (v-\k-\i);
    }
}

\foreach \i in {1,...,\N} {
    \draw[arc] (u-0-\i) -- (v-0-\i);
    \node[note, anchor=south] at ($(u-0-\i)+(0,.18)$) {$\i$};
    \node[note, anchor=north] at ($(v-0-\i)+(0,-.18)$) {$\i$};
}

\node[labelvertex] (a) at (6.6,-4.25) {$a$};
\node[note] at (6.6,-4.6) {$I_a:=[3,7]=[3,6]\cup[4,7]$};

\draw[arc] (v-2-3) -- (a);
\draw[arc] (v-2-4) -- (a);

\draw[patharc] (u-2-2) -- (u-1-4);
\draw[patharc] (u-1-4) -- (u-0-4);
\draw[patharc] (u-0-4) -- (v-0-4);
\draw[patharc] (v-0-4) -- (v-1-3);
\draw[patharc] (v-1-3) -- (v-2-3);
\draw[patharc] (v-2-3) -- (a);

\foreach \p in {u-2-2,u-1-4,u-0-4,v-0-4,v-1-3,v-2-3} {
    \node[marked] at (\p) {};
}

\node[note, anchor=south] at ($(u-2-2)+(-.05,.12)$) {$u_{[2,5]}$};

\node[note, anchor=north] at ($(v-2-3)+(0,-.18)$) {$v_{[3,6]}$};
\node[note, anchor=north] at ($(v-2-4)+(0,-.18)$) {$v_{[4,7]}$};

\node[note, anchor=south] at ($(u-1-1)+(0,.18)$) {$u_{[1,2]}$};
\node[note, anchor=west]  at ($(u-0-8)+(.18,0)$) {$u_{[8,8]}$};

\end{tikzpicture}
}
\caption{The graph $H$ with $N=8$ and a single interval $I_a = [3,7]$.
  There is a~directed path from $u_{[2,5]}$ to $a$.
  Note that, on the contrary, one cannot reach $a$ from $u_{[1,2]}$ or from $u_{[8,8]}$.}
\label{fig:interval-ds}
\end{figure}

  \newcommand{\vis}{\mathsf{vis}}
  Apart from the graph $H$, we also initialize a~visited bit $\vis_v$ for each $v \in V(H)$, initially unset.
  Then consider a~query $J$.
  As above, write $J$ as a~union of two intervals $J_1, J_2 \in \mathcal{I}_k$ for some $k \in [0, L]$.
  Then perform two graph searches in $H$ from the outvertices $u_{J_1}, u_{J_2}$, visiting only previously unvisited vertices of $H$, and marking these vertices as visited.
  For each newly visited vertex $a \in A$, we report $a$ as the label corresponding to the interval $I_a$ intersecting $J$.

  The correctness of the data structure is clear.
  For the efficiency of the data structure, observe that $H$ has size $\Oh(N \log N + |A|)$ and can be constructed in time proportional to its size.
  Then, each of the $q$ queries spawns at most two graph searches in $H$ that only visit previously unvisited vertices.
  Therefore, all queries are processed in total time $\Oh(N \log N + |A| + q)$, as required.
\end{proof}

We then show how to use \Cref{lem:intv-ds-quasilinear} to bootstrap a~data structure for tiny universes to handle an~exponentially larger universe.
In the next lemma, $b$ is the size of the machine word.

\begin{lemma}[Bootstrapping for interval intersections]
  \label{lem:intv-ds-bootstrapping}
  Let $N_{\max} < 2^b$ be an~integer.
  Suppose there exists a~data structure for \IntvDS that processes instances with universe of size $N \leq \log N_{\max}$, set of labels $A$, and $q$ queries in $\Oh(N + |A| + q)$ time.
  Then there exists a~data structure for \IntvDS that processes instances with universe of size $N \leq N_{\max}$, the set of labels $A$, and $q$ queries in $\Oh(N + |A| + q)$ time.
\end{lemma}
\begin{proof}
  We construct the promised data structure as follows.
  Let $\ell = \max(1, \left\lfloor \log N \right\rfloor)$ and partition $[N]$ into $m = \left\lceil \frac{N}{\ell} \right\rceil$ blocks $B_1, \ldots, B_m$ of length $\ell$; the last block is shorter if needed.
  For an~interval $I \subseteq [N]$, define:
  \newcommand{\inter}{\mathsf{inter}}
  \newcommand{\core}{\mathsf{core}}
  \newcommand{\fringe}{\mathsf{fringe}}
  \newcommand{\bactive}{\mathsf{active}}
  \begin{itemize}
    \item $\inter(I) \subseteq [m]$ to be the (nonempty) interval of indices of blocks that intersect $I$;
    \item $\core(I) \subseteq [m]$ to be the (possibly empty) interval of indices of blocks that are fully contained within $I$;
    \item $\fringe(I)$ to be the collection (of size at most $2$) of nonempty intervals $I \cap B_i$ over all blocks $i \in \inter(I) \setminus \core(I)$.
  \end{itemize}
  See \cref{fig:inter-core-fringe} for an illustration.
  \newcommand{\rel}{\mathsf{rel}}
  All of the above can be computed from $I$ in constant time.
  Also, whenever $I$ is a~subinterval of $B_i$, let $\rel_i(I) \subseteq [\ell]$ be the position of the interval $I$ relative to $B_i$, i.e., the interval $I$ with both endpoints shifted by $-(i - 1)\ell$.

\begin{figure}[!ht]
\centering
\begin{tikzpicture}[
    corefill/.style={fill=gray!18},
    fringefill/.style={fill=gray!45},
    block/.style={draw, line width=.35pt},
    interval/.style={line width=.9pt},
    brace/.style={
        decorate,
        decoration={brace, mirror, amplitude=4pt},
        thick
    },
    note/.style={
        font=\scriptsize,
        inner sep=1pt
    }
]

\def\nblocks{8}
\def\s{1.35}
\def\t{6.65}
\def\h{1.55}

\fill[fringefill] (\s,0) rectangle (2,1);
\fill[corefill]   (2,0) rectangle (6,1);
\fill[fringefill] (6,0) rectangle (\t,1);

\foreach \i in {1,...,\nblocks} {
    \pgfmathsetmacro{\x}{\i-1}
    \draw[block] (\x,0) rectangle ++(1,1);
    \node[note] at (\x+.5,.5) {$B_{\i}$};
}

\draw[interval] (\s,1.13) -- (\t,1.13);
\draw[interval] (\s,1.03) -- (\s,1.23);
\draw[interval] (\t,1.03) -- (\t,1.23);
\node[note, anchor=south] at ({(\s+\t)/2},1.18) {$I$};

\draw[interval] (\s,\h) -- (2,\h);
\draw[interval] (\s,\h-.10) -- (\s,\h+.10);
\draw[interval] (2,\h-.10) -- (2,\h+.10);
\node[note, anchor=south] at ({(\s+2)/2},\h+.05) {$I_{\mathrm L}$};

\draw[interval] (6,\h) -- (\t,\h);
\draw[interval] (6,\h-.10) -- (6,\h+.10);
\draw[interval] (\t,\h-.10) -- (\t,\h+.10);
\node[note, anchor=south] at ({(6+\t)/2},\h+.05) {$I_{\mathrm R}$};

\draw[brace] (1,-.1) -- (7,-.1)
    node[midway, yshift=-.34cm, note]
    {$\inter(I)=[2,7]$};

\draw[brace] (2,-.7) -- (6,-.7)
    node[midway, yshift=-.34cm, note]
    {$\core(I)=[3,6]$};

\node[note, align=center] at (4,2)
    {$\fringe(I)=\{I_{\mathrm L}, I_{\mathrm R}\}$};

\end{tikzpicture}
\caption{The objects $\inter(I)$, $\core(I)$, and $\fringe(I)$ for $m=8$ and some interval $I$.}
\label{fig:inter-core-fringe}
\end{figure}

  We instantiate three types of auxiliary data structures for \IntvDS.
  \begin{itemize}
    \item $D_{\rm core}$ holds universe $[m]$.
      For every $a \in A$ with nonempty $\core(I_a)$, it holds the interval $\core(I_a)$ with label $a$.
      $D_{\rm core}$ is implemented via \Cref{lem:intv-ds-quasilinear}.
    \item $D_{\rm fringe}$ holds universe $[m]$.
      For every $a \in A$ and every fringe interval $I' \in \fringe(I_a)$ within an $i$th block, it holds the interval $[i, i]$ with label $a$.\footnote{Technically, this may cause two intervals to be assigned to the same label $a$. Formally, in this case, we create two labels $a_1, a_2$ and assign them to the two intervals. Then, whenever $a_i$ is reported by $D_{\rm fringe}$, we translate the label $a_i$ to $a$ on the fly.} $D_{\rm fringe}$ is also implemented via \Cref{lem:intv-ds-quasilinear}.
    \item For each block $i \in [m]$, the data structure $D_i$ carries universe $[\ell]$.
      For every $a \in A$ and every fringe interval $I' \in \fringe(I_a)$ within the $i$th block, it holds the interval $\rel_i(I')$ labeled $a$.
      Each $D_i$ is implemented via the assumed linear-time data structure for small universes.
  \end{itemize}
  Moreover, for each $a \in A$, we introduce a~bit $\bactive_a$, initially set, representing whether the label $a$ is still active.
  
  Given a~query interval $J \subseteq [N]$, we perform the following queries to auxiliary data structures and gather all the returned labels.
  \begin{itemize}
    \item Query $D_{\rm core}$ with $\inter(J)$.
    \item Query $D_{\rm fringe}$ with $\core(J)$ as long as $\core(J)$ is nonempty.
    \item For each $J' \in \fringe(J)$, let $J'$ be contained within the $j$th block. Then query $D_j$ with $\rel_j(J')$.
  \end{itemize}

  For each returned label $a$, if $a$ is still active, we deactivate $a$ (i.e., unset $\bactive_a$) and add $a$ to the result label set.
  This concludes the description of the data structure; we now argue its correctness and efficiency.

  \subparagraph{Correctness.}
  Let $I_a$ be an~input interval labeled $a$ and $J$ be a~query interval.
  First, suppose $a$ is reported by the query.
  Then, $a$ was still active before the query.
  If $D_{\rm core}$ reported~$a$, then $\core(I_a) \cap \inter(J) \neq \emptyset$, which easily implies $I_a \cap J \neq \emptyset$.
  Next, if $D_{\rm fringe}$ reported $a$, then a~fringe subinterval $I' \in \fringe(I_a)$ resides in a~block belonging to $\core(J)$; hence $I' \cap J \neq \emptyset$ and so $I_a \cap J \neq \emptyset$.
  Finally, if $a$ is returned by $D_i$, then there exist $I' \in \fringe(I_a)$ and $J' \in \fringe(J)$, both within the $i$th block, so that $I' \cap J' \neq \emptyset$; hence trivially $I_a \cap J \neq \emptyset$.

  Conversely, suppose $I_a \cap J \neq \emptyset$ and $a$ is active.
  Let $x$ be any element of the intersection and suppose $x$ is within the $i$th block.
  First, if $i \in \core(I_a)$, then $D_{\rm core}$ definitely reports $a$ since $J$ intersects the $i$th block.
  Otherwise, $x \in I'$ for some $I' \in \fringe(I_a)$.
  If $i \in \core(J)$, then $D_{\rm fringe}$ reports $a$ as the data structure carries the interval $[i, i]$ labeled $a$.
  In the opposite case, $x \in J'$ for some $J' \in \fringe(J)$ and so $I' \cap J' \neq \emptyset$.
  Therefore, $D_i$ returns $a$.
  We conclude that an~active label $a$ is output if and only if $I_a \cap J \neq \emptyset$, which implies the correctness of our implementation.
  
  \subparagraph{Time.}
  Suppose $q$ queries are processed by our data structure.
  Both $D_{\rm core}$ and $D_{\rm fringe}$ maintain a~universe of size $m$, hold at most $\Oh(|A|)$ labels, and accept at most $q$ queries.
  Thus both of these auxiliary data structures initialize and process the queries in total time $\Oh(m \log m + |A| + q) = \Oh(N + |A| + q)$.
  Then, the data structures $D_i$ have total universe size $m \ell = \Oh(N)$, carry at most $\Oh(|A|)$ labels in total, and accept $\Oh(q)$ queries in total.
  By our assumption, they all process all the queries in total time $\Oh(N + |A| + q)$.
  Therefore, the time complexity of our constructed data structure is also linear.
\end{proof}
Note that in \Cref{lem:intv-ds-bootstrapping}, the constant hidden in the $\Oh$ notation of the resulting data structure is greater than the constant hidden in the prerequisite data structure.

It remains to give a~linear-time data structure for sufficiently small universe sizes.

\begin{lemma}[Efficient interval intersections for tiny universes]
  \label{lem:intv-ds-tiny-universe}
  Let $b$ be the size of the machine word and $\ell \leq \sqrt{b}$.
  There exists a~data structure for any instance of \IntvDS with the universe size at most $\ell$ that takes a~finite set $A$ of labels and processes $q$ queries in total time $\Oh(\ell + |A| + q)$.
  The data structure requires a~global table, which can be computed given only $\ell$ in time $\Oh(2^{\ell^2})$.
\end{lemma}
\begin{proof}
  Fix $\ell$ as the universe size; if the universe size of an~instance is smaller than $\ell$, we can artificially increase it to $\ell$.
  Define the \emph{state} of an~instance of \IntvDS to be the set of \emph{active} integer intervals of $[\ell]$: the intervals assigned to the active labels of the instance.
  Note that there are exactly $2^{\binom{\ell + 1}{2}}$ different states since there are $\binom{\ell + 1}{2}$ integer subintervals of $[\ell]$.
  We represent states in memory as bitmasks of length $\binom{\ell + 1}{2}$; this is permitted in our computation model as $b \geq \binom{\ell + 1}{2}$.

  In the preprocessing stage, we initialize a~global two-dimensional table $T$ as follows.
  For a~state $S$ and an~interval $J \subseteq [\ell]$, we define $T[S][J] = (S_{\rm in}, S_{\rm out})$, where $S_{\rm in}$ is the set of intervals of $S$ intersecting $J$, and $S_{\rm out}$ is the remaining intervals of $S$.
  Observe that when an~instance of \IntvDS is in state $S$, then performing a~query $J$ on it transitions it to state $S_{\rm out}$ and returns all labels with the corresponding intervals in $S_{\rm in}$.
  Let also, for a~state $S$ and an~interval $I \in S$, the value $F[S][I]$ be an~integer such that, for each $S$, the mapping from intervals $I \in S$ to the values $F[S][I]$ is a~bijection between $S$ and $[|S|]$.
  We also explicitly compute an~$\Oh(\ell^2)$-sized array $M$ mapping integer subintervals of $[\ell]$ to their bit index within the state.
  The arrays can be constructed in time $2^{\binom{\ell + 1}{2}} \ell^{\Oh(1)} \in \Oh(2^{\ell^2})$.

  \newcommand{\Sinit}{S_{\mathrm{init}}}
  We now give an~implementation of the data structure.
  On initialization, we compute the initial state $\Sinit$ from $\ell$ and~$A$ using $M$.
  Then create buckets $B_1, \ldots, B_{|\Sinit|}$ of labels from~$A$, and for each label $a \in A$, insert $a$---corresponding to an~interval $I_a \in \Sinit$---into bucket $B_{F[\Sinit][I_a]}$.
  This can be done in total $\Oh(\ell + |A|)$ time.
  Let also $S$ track the current state of the data structure; initially $S = S_{\rm init}$.
  Then, for a~given query $J$, let $(S_{\rm in}, S_{\rm out}) = T[S][J]$.
  We update the state of the data structure to $S_{\rm out}$ and return all labels of the intervals in~$S_{\rm in}$ (i.e., all labels in the buckets $B_{F[\Sinit][I]}$ for $I \in S_{\rm in}$).
  This can clearly be done in time $\Oh(1 + t)$, where $t$ denotes the total number of returned labels.
  Therefore, all queries are processed in total time $\Oh(\ell + |A| + q)$.
\end{proof}

Given \Cref{lem:intv-ds-bootstrapping,lem:intv-ds-tiny-universe}, the proof of \Cref{lem:intv-ds-linear} is direct.

\begin{proof}[Proof of \Cref{lem:intv-ds-linear}]
  Fix the universe size $N < 2^b$, and let $\ell = \max(1, \left\lceil \log \log N \right\rceil)$.
  By \Cref{lem:intv-ds-tiny-universe}, after preprocessing in time $\Oh(2^{\ell^2}) = o(N)$, we can construct a~data structure for \IntvDS processing instances with a~universe of size at most $\ell$, the set of labels~$A$, and $q$ queries in time $\Oh(\ell + |A| + q)$.
  Applying \Cref{lem:intv-ds-bootstrapping} the first time with $N_{\max} = 2^\ell$, we get a~data structure for universes of size $n \leq 2^\ell$ with total time $\Oh(n + |A| + q)$.
  Applying \Cref{lem:intv-ds-bootstrapping} once again with $N_{\max} = N$, we end with the required data structure.
\end{proof}

\section{Wrap-up of the algorithmic corollaries}

The proof of~\cref{thm:lnc-apsp} now simply wraps the results of previous sections together.

\thmlncapsp*
\begin{proof}
  By~\cref{thm:lnc-implies-bvsd}, $\mathcal C$ has bounded versatile symmetric difference.
  Thus, by~\cref{thm:sd-degen-seq-cor}, we compute an sd-degeneracy sequence of~$G$ of constant width in time $\Oh(n^2)$.
  By \cref{lem:sdd-seq-to-ibp}, we get an interval biclique partition of~$G$ with $\Oh(n)$ bicliques in further $\Oh(n \log n + \abs{E(G)})=\Oh(n^2)$ time. 
  We finally invoke~\cref{thm:sssp-ibp} from every vertex as~source, and conclude.
\end{proof}

\thmlncmatmult*
\begin{proof}
  As in the proof of~\cref{thm:lnc-apsp}, we obtain an interval biclique partition of~$G$ with $\Oh(n)$ bicliques in $\Oh(n^2)$ time.
  We conclude with~\cref{lem:matmult-ibp}.
\end{proof}

\section{Finding an sd-degeneracy sequence in randomized linear time}\label{sec:sd-deg-rand}\label{sec:randomized-algorithm}

In this section, we prove \cref{thm:LV-algo-gen}: we show how to compute an sd-degeneracy sequence of width $\Oh(n^{1-1/\rho})$ in expected linear time in graphs of VC density~$\rho$.
In particular, this yields an expected linear-time algorithm that outputs sd-degeneracy sequences of constant width in classes of linear neighborhood complexity (\cref{thm:LV-algo}).
The proof relies on a~result in learning theory and VC-combinatorics related to Haussler's packing lemma (\cref{thm:Haussler}).

A~\emph{set system} on a~domain $D$ is a~set $\cal F$ of subsets of $D$.
For $A\subseteq D$, denote \[\mathcal F\res{A} := \{F \cap A : F \in \mathcal F\}.\]
The \emph{shatter function} of a~set system $\cal F$ is the function $\pi_{\cal F} : \Nn \to \Nn$ defined by \[\pi_{\cal F}(m) := \max\{|\mathcal F\res{A}| : A \subseteq D, |A| \le m\}.\]
The \emph{VC dimension} of a set system $\cal F$ is the largest integer $d$ such that there exists a~subset $A \subseteq D$ of size $d$ with $\mathcal F\res{A} = 2^{A}$.

Given a finite domain $X$, an \emph{$s$-sample} is a~tuple $S=(x_1,\ldots,x_s)\in X^s$ of~$s$ elements of~$X$.
We will consider $s$-samples of~$X$ drawn uniformly at random, that is, each element $x_i$ of~$S$ is drawn independently and uniformly at random from~$X$.
(Note that an $s$-sample can contain the same element multiple times.)

Given a~set $F$ and an $s$-sample $S = (x_1, \dots, x_s)$, denote $F\cap S:= F\cap \set{x_1, \dots, x_s}$.
For a set system $\cal F$ on a domain $X$ and integer $s\ge 1$, denote
\[
  \cal F\res s := \{(S, F\cap S) \mid S \in X^s, F \in \mathcal F\},
\]
that is, the set of all restrictions of sets in $\cal F$ to all possible samples of size $s$.
The next lemma originates in \cite[Lemma 9]{KUPAVSKII202022} as part of a proof of Haussler's packing lemma.
It is in the spirit of the classical PAC learning model, and states that sets from a set system of bounded VC dimension can be learned from a small random sample.

\begin{lemma}[\cite{KUPAVSKII202022}]\label{thm:learning}
  Let $\mathcal F$ be a set system on a finite domain $X$ of VC dimension at~most~$d$.
  For any $\eps > 0$ and integer $s \ge 1$, there exists a~map $f : \mathcal F\res s \to 2^X$ such that for every $F \in \cal F$,
  \[
  \underset{S \sim X^s}{\mathbb{P}}\left(\abs{F\,\triangle\,f(S, F\cap S)} < \frac{d}{\eps s} \cdot \abs{X}\right) \ge 1 - \eps,
  \]
  where $S \sim X^s$ denotes that $S$ is an $s$-sample of~$X$, drawn uniformly at random.
\end{lemma}

The next algorithm relies on the existence of the map $f$ as above, but does not require to compute it explicitly.

\begin{lemma}\label{lem:sampling_pairs}
Fix real numbers $c,\rho\ge 1$ and $\delta>0$, and let $\alpha:=1-1/\rho$.
There is a Las Vegas randomized algorithm that, given an $n$-vertex graph $G$ with $(c,\rho)$-polynomial neighborhood complexity and a~set $U \subseteq V(G)$ with $|U| \ge \max(\delta n,2)$, outputs in expected time
\[
  \Oh_{c,\rho,\delta}\left(n + \sum_{u \in U} \deg_G(u)\right)
\]
a~set of $\Omega_{c,\rho,\delta}(|U|)$ disjoint pairs of $\Oh_{c,\rho,\delta}(n^\alpha)$-near-twins contained in $U$.
\end{lemma}
\begin{proof}
We may assume that $\delta\le 1$.
If $|U|<2^{\rho+1}c$, then $n\le |U|/\delta=\Oh_{c,\rho,\delta}(1)$.
An arbitrary maximal pairing of~$U$ has size at least $|U|/4$, and every pair has symmetric difference at most $n=\Oh_{c,\rho,\delta}(n^\alpha)$, so we are done.
Henceforth assume that $|U|\ge 2^{\rho+1}c$.

We first consider the case when $U$ contains no pair of vertices which are twins in~$G$.
Let
\[
  \mathcal F := \{N_G(u) \mid u \in U\}
\]
be the set system with domain $V(G)$ of neighborhoods of vertices of~$U$.
The members of $\cal F$ are distinct, and, since $G$ has $(c,\rho)$-polynomial neighborhood complexity,
\[
  \pi_{\cal F}(m) \le \pi_G(m) \le c\cdot m^\rho
\]
for every positive integer $m$.
Let $d_0$ be the VC dimension of $\cal F$.
Since $|\cal F|=|U|\ge 2$, we have $d_0\ge 1$, and
\[
  2^{d_0}\le \pi_{\cal F}(d_0)\le c\cdot d_0^\rho.
\]
Consequently, there is an integer $d := d(c,\rho)$, fixed independently of the input, such that $d_0\le d$.

Let
\[
  s := \left\lfloor \left(\frac{|U|}{2c}\right)^{1/\rho} \right\rfloor,
  \qquad \eps := \frac{1}{16},
  \qquad k := \left\lceil \frac{2dn}{\eps s} \right\rceil.
\]
Our lower bound on $|U|$ ensures that
\[
  s \ge \frac{1}{2}\left(\frac{|U|}{2c}\right)^{1/\rho}
  \qquad\text{and}\qquad
  cs^\rho \le \frac{|U|}{2}.
\]
Moreover,
\[
  k \le 1 + \frac{4dn(2c)^{1/\rho}}{\eps |U|^{1/\rho}}
  \le \Oh_{c,\rho,\delta}\left(n^{1-\frac{1}{\rho}}\right).
\]

Fix a map $f$ as in \Cref{thm:learning}, using the upper bound $d$ on the VC dimension.
For an $s$-sample~$S$, we call an element $F \in \mathcal F$ \emph{bad} (with respect to~$S$) if
\[
  \abs{F\,\triangle\,f(S, F\cap S)} \ge \frac{k}{2}.
\]
By \Cref{thm:learning}, for any fixed $F\in \mathcal F$, the probability over the random choice of~$S$ that $F$ is bad is at~most~$\eps$.
Thus, the expected number of bad elements of $\cal F$ is at~most~$\eps |U|$.
By Markov's inequality,
\[
  \underset{S \sim V(G)^s}{\mathbb{P}}
  \left(\left|\left\{F \in \mathcal{F} \mid F\text{ is bad with respect to }S\right\}\right| > |U|/8\right)
  \le \frac{\eps |U|}{|U|/8} = 8\eps = \frac{1}{2}.
\]

One attempt of the algorithm draws an $s$-sample $S$ of $V(G)$ and lets $\widehat S\subseteq V(G)$ be its underlying set.
Consider the set system
\[
  \mathcal F\res{\widehat S} := \{N_G(u) \cap \widehat S \mid u \in U\}.
\]
We have
\[
  |\mathcal F\res{\widehat S}|
  \le \pi_{\cal F}(|\widehat S|)
  \le c|\widehat S|^\rho
  \le cs^\rho
  \le \frac{|U|}{2}.
\]
For each $\trF \in \mathcal F\res{\widehat S}$, let the \emph{bucket $B(\trF)$} be the set of sets $F \in \cal F$ such that $F \cap \widehat S = \trF$.
Thus, the buckets partition $\cal F$.

Inside every bucket, we arbitrarily choose a maximum set of pairwise disjoint pairs, leaving at most one set unpaired.
Let $\pairs$ be the collection of pairs obtained from all the buckets, and let $\pairs_{\mathrm{good}}\subseteq\pairs$ consist of those pairs $\{F_1,F_2\}$ such that $|F_1\triangle F_2|\le k$.
If $|\pairs_{\mathrm{good}}|\ge |U|/8$, we output the corresponding pairs of vertices in~$U$; otherwise, we discard the attempt and resample~$S$.

\begin{claim}\label{clm:one-attempt-time}
  One attempt takes expected time $\Oh(n + \sum_{u \in U} \deg_G(u))$.
\end{claim}
\begin{claimproof}
  Sampling~$S$ takes $\Oh(n)$ time.
  Computing $\mathcal F\res{\widehat S}$ and partitioning $\mathcal F$ into buckets takes expected time
  $\Oh(n + \sum_{u\in U}\deg_G(u))$ by scanning the adjacency lists of the vertices in~$U$.
  Building the candidate pairs takes $\Oh(|U|)$ time.

  Given a candidate pair $\{F_1,F_2\}$, computing $F_1\triangle F_2$ takes
  $\Oh(|F_1|+|F_2|)$ time.
  As the candidate pairs are disjoint, checking all of them takes
  $\Oh(|U| + \sum_{u\in U}\deg_G(u))$ time.
\end{claimproof}

\begin{claim}
  Every attempt succeeds with probability at least $1/2$.
\end{claim}
\begin{claimproof}
Let $F_1,F_2$ lie in the same bucket.
Then $F_1\cap S=F_2\cap S$, and hence
\[
  \abs{F_1 \triangle F_2}
  \le \abs{F_1 \triangle f(S,F_1\cap S)}
  + \abs{F_2 \triangle f(S,F_2\cap S)}.
\]
Therefore, every pair in $\pairs\setminus\pairs_{\mathrm{good}}$ contains a bad element of $\cal F$.
As the candidate pairs are disjoint, the number of such pairs is at most the number of bad elements.

If the buckets are $B_1,\ldots,B_t$, then $t = |\mathcal F\res{\widehat S}| \le |U|/2$ and
\[
  |\pairs| = \sum_{j=1}^t \left\lfloor \frac{|B_j|}{2} \right\rfloor
  \ge \frac{|U|-t}{2}
  \ge \frac{|U|}{4}.
\]
With probability at least $1/2$, at most $|U|/8$ elements of $\cal F$ are bad, and on that event
\[
  |\pairs_{\mathrm{good}}|
  \ge |\pairs|-\frac{|U|}{8}
  \ge \frac{|U|}{8}.
\]
Thus, the attempt succeeds.
\end{claimproof}

The expected number of attempts is at most two, so \cref{clm:one-attempt-time} proves the claimed expected running time when $U$ contains no pair of twins in~$G$.

\medskip

In the general case, construct a graph $G'$ by adding, for every $u\in U$, a new private pendant vertex $x_u$ adjacent only to~$u$.
Then the vertices of~$U$ have pairwise distinct neighborhoods in~$G'$.
We claim that $G'$ has $(c+2,\rho)$-polynomial neighborhood complexity.
Indeed, fix a nonempty set $A\subseteq V(G')$, and put $A_0:=A\cap V(G)$ and $A_1:=A\setminus V(G)$.
The traces of the old vertices are determined by their traces on~$A_0$, except that at most $|A_1|$ of them can acquire a private pendant element.
The new pendant vertices have traces among $\emptyset$ and the singletons $\{u\}$ for $u\in A_0$.
Consequently, the total number of traces on~$A$ is at most
\[
  c|A|^\rho+|A|+1\le (c+2)|A|^\rho.
\]

Also, $|V(G')|\le 2n$, $|U|\ge (\delta/2)|V(G')|$, and $\deg_{G'}(u)=\deg_G(u)+1$ for every $u\in U$.
We can therefore apply the twin-free case to $G'$ and~$U$ with constants $c+2$, $\rho$, and $\delta/2$.
This application has the claimed expected running time.
Finally, deleting the new pendant vertices cannot increase a symmetric difference, so the returned pairs satisfy the asserted bound in~$G$ as well.
\end{proof}

We now present the randomized linear-time algorithm computing an sd-degeneracy sequence.

\csname thmlvalgo-gen\endcsname*
\begin{proof}
Let $c\ge 1$ be such that every graph in $\mathcal C$ has $(c,\rho)$-polynomial neighborhood complexity.
Write $n:=|V(G)|$, $m:=|E(G)|$, and $\alpha:=1-1/\rho$.
We build a~sequence $G_1,G_2,\ldots$ of induced subgraphs of $G$, with $G_1:=G$.
Every $G_i$ still has $(c,\rho)$-polynomial neighborhood complexity.

In every iteration $i$ with $n_i:=|V(G_i)|\ge 3$, let $U_i\subseteq V(G_i)$ be the set of $\lceil n_i/2\rceil$ vertices of $G_i$ with smallest degrees, breaking ties arbitrarily.
By \Cref{lem:sampling_pairs}, applied to $G_i$, $U_i$, and $\delta=1/2$, in expected time
\[
  \Oh_{c,\rho}\left(n_i + \sum_{u\in U_i}\deg_{G_i}(u)\right)
\]
we find a~set $\pairs_i$ of $\Omega_{c,\rho}(n_i)$ disjoint pairs of $\Oh_{c,\rho}(n_i^\alpha)$-near-twins inside~$U_i$.
We order every pair of~$\pairs_i$ arbitrarily, append the ordered pairs to the output sequence, and obtain~$G_{i+1}$ by deleting the first vertex of every pair.
When at most two vertices remain, we append their pair if there are two vertices, and terminate.

\begin{claim}\label{clm:sdd-seq}
The output is an sd-degeneracy sequence of~$G$ of width $\Oh_{c,\rho}(n^\alpha)$.
\end{claim}
\begin{claimproof}
Consider a pair $(u,v)$ from a batch $\pairs_i$ when it is reached in the output sequence.
Earlier pairs of the same batch are vertex-disjoint from $(u,v)$, and the graph then present is an induced subgraph of~$G_i$.
Hence, its symmetric difference is no larger than it was in~$G_i$, where $(u,v)$ is a pair of $\Oh_{c,\rho}(n_i^\alpha)$-near-twins.
Since $n_i\le n$, this is $\Oh_{c,\rho}(n^\alpha)$.
If the algorithm terminates with two vertices, their symmetric difference is at most~$2$, and hence at most~$2n^\alpha$.
\end{claimproof}

\begin{claim}\label{clm:expected-time}
  The expected running time is $\Oh_{c,\rho}(n+m)$.
\end{claim}
\begin{claimproof}
We maintain $\deg_{G_i}(v)$ for every current vertex $v$, initialized in $\Oh(n+m)$ time.
When vertices are deleted, we update the degrees of their surviving neighbors; all such updates cost $\Oh(n+m)$ in total.
In iteration $i$, a~linear-time selection algorithm applied to these maintained degrees finds the $\lceil n_i/2\rceil$ vertices of smallest degree in $\Oh(n_i)$ time.

By conditional linearity of expectation, the remaining expected work is bounded by
\begin{align}\label{eq:time}
  \Oh_{c,\rho}\left(
    \sum_i n_i + \sum_i\sum_{u\in U_i}\deg_{G_i}(u)
  \right).
\end{align}

Every iteration removes $\Omega_{c,\rho}(n_i)$ vertices.
Thus, the sequence $(n_i)_i$ decreases geometrically, and
\[
  \sum_i n_i = \Oh_{c,\rho}(n).
\]
For the second term of~\eqref{eq:time}, we have
\[
  \sum_i\sum_{u\in U_i}\deg_{G_i}(u)
  = \sum_{r\ge 1}|S_r|,
  \qquad
  \text{with}~S_r:=\{(i,u)\mid u\in U_i\text{ and }\deg_{G_i}(u)\ge r\}.
\]
For every $r$ with $S_r\neq\emptyset$, let $i_r$ be the smallest index such that $U_{i_r}$ contains a vertex of degree at least~$r$.
The geometric decrease and $|U_i|=\Theta(n_i)$ imply
\[
  |S_r|\le \sum_{i\ge i_r}|U_i|=\Oh_{c,\rho}(|U_{i_r}|).
\]

For every $r\ge 1$, let
\[
  A_r:=\{v\in V(G)\mid \deg_G(v)\ge r\}.
\]
Because $U_{i_r}$ consists of the lower half of the degree order and contains a vertex of degree at least~$r$, every vertex of $V(G_{i_r})\setminus U_{i_r}$ has degree at least~$r$ in~$G_{i_r}$, and hence also in~$G$.
Therefore,
\[
  \left\lfloor \frac{n_{i_r}}{2} \right\rfloor \le |A_r|
  \qquad\text{and consequently}\qquad
  |U_{i_r}|\le |A_r|+1.
\]
Since $S_r=\emptyset$ for $r\ge n$ and $\sum_{r\ge 1}|A_r|=2m$, we obtain
\[
  \sum_i\sum_{u\in U_i}\deg_{G_i}(u)
  = \sum_{r=1}^{n-1}|S_r|
  \le \Oh_{c,\rho}\left(\sum_{r=1}^{n-1}(|A_r|+1)\right)
  = \Oh_{c,\rho}(n+m).
\]
Substituting both estimates into~\eqref{eq:time} proves the claim.
\end{claimproof}
We conclude by~\cref{clm:sdd-seq,clm:expected-time}.
\end{proof}

\section{Finding Small Patterns}

We next design a~randomized linear-time algorithm that returns a~triangle (if one exists) in graphs from a~class of linear neighborhood complexity.
The main idea is to use an sd-degeneracy sequence of the input graph to guide the search.
In the next three (almost-)linear algorithms we assume that the edges of the input graph are preprocessed into a~static dictionary in $\Oh(n+m)$ expected time and $\Oh(n+m)$ space, supporting adjacency queries in $\Oh(1)$ expected time.

\csname triangle-detection\endcsname*

\begin{proof}
  Let $G$ be any input graph.
  By~\cref{thm:LV-algo}, we get an sd-degeneracy sequence $(u_1,v_1), \ldots,$ $(u_{n-1}, v_{n-1})$ of~$G$ of width $d = \Oh_{\mathcal C}(1)$ in randomized time $\Oh_{\mathcal C}(n+m)$.

  For every $i \in [n-1]$, we define $G_i := G - \{u_1, \ldots, u_{i-1}\}$.
  Note that to find a~triangle in~$G = G_1$ (if any), it is enough to return a~triangle of~$G_1$ that contains $u_1$ or obtain the guarantee that \emph{not} all the triangles of~$G_1$ contain $u_1$, and recursively apply this strategy to~$G_2$.
  Indeed, at each step we either find a~triangle in $G_i$ or move to $G_{i+1}$ with the promise that $G_{i+1}$ contains a~triangle if and only if $G_i$ contains one.

  We set \[A_i := N_{G_i}(u_i) \setminus N_{G_i}(v_i)~~\text{and}~~X_i := N_{G_i}(u_i) \cap N_{G_i}(v_i)\]
  for every $i \in [n-1]$.

  For $i$ going from 1 to $n-2$, the algorithm proceeds as follows.
  Compute $A_i$ in time $\Oh(\abs{N_{G_i}(u_i)})$ by filtering out the vertices of $N_{G_i}(u_i)$ that are adjacent to~$v_i$.
  By assumption, $\abs{A_i} \leqslant d$.
  For every $a \in A_i$ and every $b \in N_{G_i}(u_i)$ check if $ab \in E(G_i)$.
  If so, return the triangle $u_i a b$ (break from the for loop, and terminate).
  Even if all the checks are unsuccessful, this takes $\Oh\left((d+1) \cdot \abs{N_{G_i}(u_i)}\right)$ time.
  If no triangle is found, remove $u_i$ from $G_i$ in time $\Oh(\abs{N_{G_i}(u_i)})$ to obtain $G_{i+1}$ (and move to the next iteration of the for loop).
  
  \subparagraph*{Correctness.}
  If there is a~triangle in $G_i$ that contains $u_i$ but the algorithm does not return one at the $i$th iteration of the for loop, it is of the form $u_i x y$ with $\{x,y\} \subseteq X_i$.
  Then $G_{i+1}$ still contains at~least the triangle $v_i x y$.

  \subparagraph*{Running time.}
  Computing the sd-degeneracy sequences takes randomized time $\Oh_{\mathcal C}(n+m)$.
  The rest of the algorithm takes time \[\Oh \left(\sum\limits_{i \in [n-2]} \left(d+1\right) \cdot \left(\abs{N_{G_i}(u_i)}+1\right) \right) = \Oh_{\mathcal C}(n+m).\qedhere\]
\end{proof}

We now build on the scheme of the previous proof to detect a~4-vertex clique.
The randomized almost linear algorithm relies on both an sd-degeneracy sequence \emph{and} a~Welzl order.

\csname k4-detection\endcsname*

\begin{proof}
  We obtain an sd-degeneracy sequence $(u_1,v_1), \ldots,$ $(u_{n-1}, v_{n-1})$ of the input graph $G$ of width $d = \Oh_{\mathcal C}(1)$ in randomized $\Oh_{\mathcal C}(n+m)$ time by~\cref{thm:LV-algo}, or deterministic $\Oh(n^2)$ time by \cref{thm:sdd-seq}.
  For every $i \in [n-1]$, we define $G_i, A_i, X_i$ as in the proof of~\cref{thm:triangle-detection}.

  Our general strategy remains the same.
  Now, for every $a \in A_i$, we do not merely need to know if $Z_{i,a} \neq \emptyset$ but rather whether $G_i[Z_{i,a}]$ contains an edge, where \[Z_{i,a} := N_{G_i}(u_i) \cap N_{G_i}(a).\]
  We start describing the deterministic $\Oh(n^2)$-time algorithm, which naturally uses \cref{thm:sdd-seq} as opening step.

  Let $M$ be an arbitrary adjacency matrix of $G$, on which we compute the data structure of~\cref{thm:lnc-matmult} (associated sparse matrix) in $\Oh_{\mathcal C}(n^2)$ time. 
  For $i$ going from 1 to $n-3$, the algorithm proceeds as follows.
  It computes $A_i$ in time $\Oh(\abs{N_{G_i}(u_i)})$.
  For every $a \in A_i$, it computes the indicator (column) vector $z_{i,a}$ of $Z_{i,a}$ in time $\Oh(n)$.

  We observe that \[z_{i,a}^T M z_{i,a} = 2\abs{E(G_i[Z_{i,a}])},\] and compute the matrix-vector product $y_{i,a} := M z_{i,a}$ in time $\Oh_{\mathcal{C}}(n)$, by~\cref{thm:lnc-matmult}.
  We now check in $\Oh(n)$ time whether there is a~$j \in [n]$ such that both $z_{i,a}[j]$ and $y_{i,a}[j]$ are nonzero.
  If this happens, we know that there is at least one edge in $G_i[Z_{i,a}]$ incident to the $j$th vertex, say $x$, of~$G$.
  We can find the other endpoint of such an edge, say $y$, in $\Oh(n)$ time.
  We return the 4-vertex clique $\{u_i, a, x, y\}$ (and terminate).
  
  If no $K_4$ is found, we remove $u_i$ from $G_i$ in time $\Oh(\abs{N_{G_i}(u_i)})$ to obtain $G_{i+1}$ (and move to the next iteration of the for loop).
  This concludes the algorithm.

  \medskip

  Each iteration takes $\Oh((d+1) n) = \Oh_{\mathcal C}(n)$ time because $\abs{A_i} \leqslant d$.
  The preprocessing of~$M$ takes $\Oh_{\mathcal C}(n^2)$ time.
  Therefore, the overall running time is $\Oh_{\mathcal C}(n^2)$.
  The correctness follows as in~the proof of~\cref{thm:triangle-detection}.

  \medskip

  We now describe the randomized $\Oh_{\mathcal C}(n \log^5 n + m \log n)$-time algorithm, which first calls \cref{thm:LV-algo}.
  In addition to the sd-degeneracy sequence, we will rely on a~Welzl order for~$G$; more specifically, in time $\Oh_{\mathcal C}((n+m) \log n)$ we compute a~vertex ordering $\prec: w_1w_2 \ldots w_n$ such that the neighborhood of every vertex of~$G$ is the union of $\Oh_{\mathcal C}(\log^2 n)$ $\prec$-intervals~\cite{DreierK26}.
  We further compute, for every $v \in V(G)$, the $\prec$-intervals $I_1(v), \ldots, I_{h_v}(v)$ with $h_v = \Oh_{\mathcal C}(\log^2 n)$ such that $N_G(v) = \biguplus_{j \in [h_v]} I_j(v)$; overall this takes time $\Oh(n+m)$.
  We observe that the orderings $u_1 u_2 \ldots$ and $w_1 w_2 \ldots$ are a~priori unrelated. 
  For every $v \in V(G)$, we denote by $\pi(v)$ the integer $i$ such that $v = w_i$.

  We create the following planar point set $P$, with $\abs{P} = 2m$.
  For every $uv \in E(G)$, we add points $(\pi(u),\pi(v))$ and $(\pi(v),\pi(u))$ to~$P$. 
  In $\Oh(m \log n)$ preprocessing time, a~data structure $\mathcal D_P$ can be computed that answers rectangle query (given an axis-parallel rectangle $R$, yield a~point in $R \cap P$) in time $\Oh(\log n)$~\cite[Theorem 2]{Chazelle88}.

  Let \[\widehat{Z}_{i,a} := N_G(u_i) \cap N_G(a),\] and note the $G$ subscripts (so in particular $\widehat{Z}_{i,a} \supseteq Z_{i,a}$).
  We proceed as the deterministic $\Oh_{\mathcal C}(n^2)$-time algorithm, except we decide if $G[\widehat{Z}_{i,a}]$ (not $G_i[Z_{i,a}]$) has an edge with the $\prec$-intervals $I_j(v)$ and $\mathcal D_P$.
  (We do not compute $M$ nor $z_{i,a}$.)
  We compute $\widehat{Z}_{i,a}$ as the intersection of \[\biguplus_{j \in [h_{u_i}]} I_j(u_i)~\text{and}~\biguplus_{j \in [h_a]} I_j(a).\]
  This intersection is itself the union of $\Oh_{\mathcal C}(\log^2 n)$ $\prec$-intervals, say $J_1, \ldots, J_p$, computed in $\Oh_{\mathcal C}(\log^2 n)$ time.
  For every $j, k \in [p]$, we test if $E_G(J_j, J_k)$ is nonempty with the query rectangle $J_j \times J_k$ on~$\mathcal D_P$.
  If so, we return the corresponding 4-vertex clique, and terminate.
  This test takes $\Oh_{\mathcal C}(\log^5 n)$ time in total, hence the claimed $\Oh_{\mathcal C}(n \log^5 n + m \log n)$ running time.
\end{proof}

We can go one step further and find 5-vertex cliques in randomized almost linear time.

\csname k5-detection\endcsname*

\begin{proof}
  The proof exactly follows the second part of that of~\cref{thm:k4-detection}.
  We now need to detect a~triangle of~$G$ between a~triple $J_h, J_j, J_k$ of $\prec$-intervals of $\widehat{Z}_{i,a} = N_G(u_i) \cap N_G(a)$ with $h, j, k \in [p]$.
  Instead of the planar point set $P$, we build the following integral point set $Q$ in $\mathbb Z^4$.
  To simplify the presentation, we assume that $V(G) = [n]$ and that $1, 2, \ldots, n$ is the Welzl order computed in time $\Oh_{\mathcal C}((n+m) \log n)$ such that the neighborhood of every vertex of~$G$ is the union of $\Oh_{\mathcal C}(\log^2 n)$ $\prec$-intervals~\cite{DreierK26}.
  
  For every $uv \in E(G)$, let \[N_G(u) \cap N_G(v) = [\ell_1, r_1] \uplus \ldots \uplus [\ell_s, r_s]\] where each $[\ell_t, r_t]$ is an $\prec$-interval and $s = \Oh_{\mathcal C}(\log^2 n)$, and add the point \[q_{u, v, t} := (u, v, \ell_t,- r_t)\] to~$Q$, for each $t \in [s]$.
  We build in time $\Oh(\abs{Q} \log^3 \abs{Q})=\Oh_{\mathcal C}(m \log^5 n)$ a~4-dimensional orthogonal range data structure~$\mathcal D_Q$ that, given a~query box, returns in time $\Oh(\log^3 \abs{Q})=\Oh_{\mathcal C}(\log^3 n)$ one point of $Q$ contained in the box, if one exists~\cite[Theorem 5.11]{BergCKO08}.

  Now for every triple $J_h$, $J_j$, and $J_k = [\alpha_k, \beta_k]$, we query the box $J_h \times J_j \times (-\infty,\beta_k] \times (-\infty,-\alpha_k]$ on~$\mathcal D_Q$.
      If the query returns a~point $q_{u, v, t} \in Q$, it holds that $u \in J_h$, $v \in J_j$, $uv \in E(G)$, $\ell_t \leqslant \beta_k$, and $r_t \geqslant \alpha_k$.
      In particular, the interval $[\ell_t,r_t]$ of $N_G(u) \cap N_G(v)$ intersects $J_k$.
      Let $z$ be any vertex in $[\ell_t,r_t] \cap J_k$.
      It holds that $u v z$ is a~triangle in $\widehat{Z}_{i,a} \supseteq J_k$, and we return the 5-vertex clique $\{u_i, a, u, v, z\}$.
      If no $K_5$ is found, $u_i$ is removed and the algorithm continues with~$G_{i+1}$ (next iteration).

      Each iteration takes time $\Oh_{\mathcal C}((\log^2 n)^3 \log^3 n)=\Oh_{\mathcal C}(\log^9 n)$, hence the overall running time $\Oh_{\mathcal C}(n \log^9 n + m \log^5 n)$.
\end{proof}  

One should, however, not expect an almost-linear dependence on $n$ uniformly over all clique sizes~$k$.
The Exponential Time Hypothesis (ETH) asserts that there exists a~constant $\lambda>0$ such that \textsc{3-SAT} on $n$ variables cannot be solved in time~$O(2^{\lambda n})$~\cite{Impagliazzo01}.

\begin{theorem}\label{thm:k-clique-lb}
  There is a~class $\mathcal C$ of linear neighborhood complexity such that \textsc{$k$-Clique} cannot be solved in time $f(k)\,n^{o(k/\log k)}$ for any function $f$, unless the ETH fails.
\end{theorem}

\begin{proof}
By \cite[Theorem~1.3]{Marx10} specialized to bipartite 3-regular graphs, as stated in \cite[Theorem~1.4]{Karthik24}, unless the ETH fails, \textsc{2-CSP} with alphabet $\Sigma$ and $q$ constraints cannot be solved in time $g(q)\,|\Sigma|^{o(q/\log q)}$ for any computable function $g$.

Let $\Gamma$ be such an instance, and $H = (X, Y; E)$ be its constraint graph.
Thus $\abs{E}=q$.
For every constraint $e = uv \in E$, where $u \in X$ and $v \in Y$, let $C_e \subseteq \Sigma^2$ be the set of pairs allowed by $e$.
We construct the cell gadget from the proof of~\cite[Theorem~3]{BonnetDD26} with $r=3$ and tile set $C_e$: each of its five \emph{blocks} contains a~copy of~$(a,b)$ for every $(a,b) \in C_e$.
Denote the two border blocks at the ends of its horizontal (resp.~vertical) path by $L_e$ and $R_e$ (resp.~by $U_e$ and $D_e$), and the central block by~$M_e$.
Each block is a~clique, and every copy of~$(a,b)$ in $M_e$ is adjacent to every vertex in $L_e \cup R_e \cup U_e \cup D_e$ except the four copies of~$(a,b)$.

For every $u \in X$, let $e_1, e_2, e_3$ be the edges of~$H$ incident to~$u$.
For every $i \in [3]$, we add edges between $R_{e_i}$ and $L_{e_{{i \bmod 3}+1}}$, making two vertices adjacent precisely when their first coordinates differ.
Similarly, for each variable $v \in Y$, we arbitrarily list the edges of~$H$ incident to~$v$ as $f_1, f_2, f_3$ and, for every $i \in [3]$, add edges between $D_{f_i}$ and $U_{f_{{i \bmod 3}+1}}$, making two vertices adjacent precisely when their second coordinates differ.

Denote the resulting graph by $G_\Gamma$.
By the correctness argument of \cite{BonnetDD26}, $\Gamma$ is satisfiable if and only if $\alpha(G_\Gamma)\geq 5q$, and
\[|V(G_\Gamma)| = 5\sum_{e\in E}|C_e| \leq 5q|\Sigma|^2.\]
Moreover, by the radius-2 merge sequence from \cite{BonnetDD26}, $\mw_2(G_\Gamma) = \mw_{r-1}(G_\Gamma) \leq 4r-3 = 9$.
Observe indeed that the upper bound on the width of the merge sequences does not depend on which pairs of border blocks from distinct cells are connected.
Let $\mathcal C$ consist of the complements of all graphs $G_\Gamma$.
By \cite[Theorem~1.5]{Bonamy25}, the graphs $G_\Gamma$ form a~class of linear neighborhood complexity.
This property is preserved under complementation. 

Now an algorithm for $k$-\textsc{Clique} on $\mathcal C$ running in time $f(k)n^{o(k/\log k)}$, applied to $\overline{G_\Gamma}$ with $k=5q$, would decide $\Gamma$ in time
\[f(5q)\bigl(5q|\Sigma|^2\bigr)^{o(q/\log q)} = g(q)\,|\Sigma|^{o(q/\log q)}.\]
The polynomial time needed to construct $G_\Gamma$ is absorbed in this bound.
This contradicts the ETH lower bound of~\textsc{2-CSP}.
\end{proof}

\paragraph{AI Disclosure.} 
A number of proofs in this work have been developed with the assistance of ChatGPT 5.5 Pro. 
\begin{itemize}[nosep]
  \item The idea to use a~longest-common-extension data structure, critical to \cref{thm:sdd-seq}, has been brought up by ChatGPT.
  \item The details of the bootstrapping process presented in \cref{lem:intv-ds-bootstrapping} and \cref{lem:intv-ds-tiny-universe} were developed with ChatGPT's assistance. (The data structure of \cref{lem:intv-ds-quasilinear} was developed by the authors.)
  \item A number of technical details in the proof of \cref{thm:LV-algo} have been developed with ChatGPT's assistance.
\end{itemize}
The suspected generalizations in~\cref{thm:sdd-seq-gen,thm:LV-algo-gen} were confirmed by GPT-5.6 Sol Pro, which also found the extensions in~\cref{thm:k4-detection,thm:k5-detection} of the triangle-detection algorithm. 
In each case, while ChatGPT suggested pieces of the write-up, the final write-up is due to the authors and has been polished and verified by us.

\newpage

\appendix

\section{Neighborhood complexity and versatile symmetric difference}

We give a~shorter (albeit non-algorithmic) proof for the inclusion of class families that can already be deduced in~\cref{lem:sampling_pairs}. 
In \Cref{thm:lnc-implies-bvsd}, we derive versatile symmetric difference from VC density thanks to the following fundamental result of Haussler.

\begin{theorem}[Haussler's packing lemma~\cite{Haussler95}]\label{thm:Haussler}
  Let $\cal F$ be a set system of VC dimension $d_0$ with $|\mathcal F| > 1$ on a finite domain~$D$, and let $c,\rho\ge 1$ be real numbers such that $\pi_{\cal F}(m) \le c\cdot m^\rho$ for every positive integer~$m$.
  Then $d_0 \le \Oh(\log c+\rho\log \rho)$, and there are distinct $F,F'\in\cal F$ such that
  \[
    |F\triangle F'|
    \le d_0\cdot c^{1/\rho}\cdot\frac{|D|}{|\mathcal F|^{1/\rho}}
    = \Oh_{c,\rho}\left(\frac{|D|}{|\mathcal F|^{1/\rho}}\right).
  \]
\end{theorem}

\begin{proof}
  Since $|\cal F|>1$, we have $d_0\ge 1$.
  Hence,
  \[
    2^{d_0} \le \pi_{\cal F}(d_0)\le c\cdot d_0^\rho,
  \]
  which implies that $d_0\le \Oh(\log c+\rho\log\rho)$.

  Let
  \[
    \delta:=\min\{|F \triangle F'| : F,F' \in \mathcal F, F \neq F'\}.
  \]
  By Theorem 1 of~\cite{Haussler95}, more precisely by the proof\footnote{The constant $C_2$ on page 158 of \cite{matousek1999geometric} is $c\cdot d_0^\rho$ in our notation, and the last inequality on page 159, for $\cal P=\cal F$, gives $|\cal F|=|\mathcal P|\le c\cdot(d_0|D|/\delta)^\rho$.} of \cite[Lemma 5.14]{matousek1999geometric},
  \[
    |\mathcal F| \le c\cdot\left(\frac{d_0|D|}{\delta}\right)^\rho.
  \]
  Rearranging gives
  \[
    \delta\le d_0\cdot c^{1/\rho}\cdot\frac{|D|}{|\mathcal F|^{1/\rho}},
  \]
  as required.
\end{proof}

We can now prove the main result of this section.
Note that a~stronger algorithmic result is proved in \Cref{lem:sampling_pairs}.

\begin{theorem}\label{thm:lnc-implies-bvsd}
  For every fixed pair of real numbers $c,\rho\ge 1$, there is a constant $d := d(c,\rho)$ such that every $n$-vertex graph $G$ with $(c,\rho)$-polynomial neighborhood complexity contains at least $\lfloor n/8\rfloor$ disjoint pairs of vertices $u,v$ satisfying $\sd_G(u,v)\le d n^{1-1/\rho}$.
 
  Consequently, every hereditary graph class of VC density~$\rho$ has versatile symmetric difference of exponent $1-1/\rho$.
  In particular, every class of linear neighborhood complexity has bounded versatile symmetric difference.
\end{theorem}
\begin{proof}
Fix an $n$-vertex graph $G$ with $(c,\rho)$-polynomial neighborhood complexity.
Let $U\subseteq V(G)$ be a largest subset of vertices without twins in~$G$, that is, $N_G(u)\neq N_G(u')$ for any distinct $u,u'\in U$.
For each $u\in U$, let $C(u)$ be the set of vertices in $V(G)$ that are twins of~$u$.
Then $\{C(u):u\in U\}$ forms a partition of $V(G)$.
We split the argument depending on whether $|U|\le n/2$.

\medskip

\noindent{\bf Case $\bm{|U|\le n/2}$:}
For every $u\in U$, pairing vertices arbitrarily inside $C(u)$ gives $\floor{|C(u)|/2}$ disjoint pairs of twins.
Thus, the total number of pairs is at least
\[
  \sum_{u\in U} \left\lfloor \frac{|C(u)|}{2} \right\rfloor
  \ge \frac{1}{2}\sum_{u\in U}(|C(u)|-1)
  = \frac{n-|U|}{2}
  \ge \frac{n}{4}.
\]
In particular, this gives at least $\lfloor n/8\rfloor$ pairs satisfying the claimed bound.

\medskip

\noindent{\bf Case $\bm{|U|>n/2}$:}
Consider the set system of neighborhoods
\[
  \mathcal F_1:=\{N_G(u)\mid u\in U\}
\]
over the domain $V(G)$.
Its members are distinct and $|\mathcal F_1|=|U|>n/2$.
Every subsystem $\mathcal F'\subseteq\mathcal F_1$ satisfies
\[
  \pi_{\mathcal F'}(m)\le \pi_G(m)\le c\cdot m^\rho
\]
for every positive integer~$m$.
By \Cref{thm:Haussler}, there is a~constant $k := k(c,\rho)$ such that every subsystem $\mathcal F'$ with at least two members contains distinct $F,F'$ with
\[
  |F\triangle F'|\le k\cdot\frac{n}{|\mathcal F'|^{1/\rho}}.
\]

The neighborhood $N_G(u)$ will be referred to as $F^u$ when considered as a member of the set system.
Starting with $\mathcal F_1$, repeat $\lfloor n/8\rfloor$ times: choose such a pair $F^{u_i},F^{v_i}$ and remove both sets to obtain $\mathcal F_{i+1}$.
Before every choice, fewer than $n/4$ sets have been removed, so the current subsystem has more than $n/4$ members and the preceding application of \Cref{thm:Haussler} applies.
The corresponding vertex pairs $\{u_i,v_i\}$ are disjoint, and every chosen pair satisfies
\[
  |N_G(u_i)\triangle N_G(v_i)|
  \le k\cdot\frac{n}{|\mathcal F_i|^{1/\rho}}
  < k\cdot 4^{1/\rho}n^{1-1/\rho}.
\]

Taking $K:=k\cdot4^{1/\rho}$ proves the first assertion.
For the consequence, if $\mathcal C$ has VC density~$\rho$, choose a constant $c$ witnessing this fact and apply the first assertion to every graph in~$\mathcal C$.
With $d:=\max(8,\lceil K\rceil)$, each $n$-vertex graph in $\mathcal C$ has at least $\lfloor n/d\rfloor$ disjoint pairs with symmetric difference at most $dn^{1-1/\rho}$.
Thus, $\mathcal C$ has versatile symmetric difference of exponent $1-1/\rho$.
\end{proof}

\end{document}